\documentclass[draftcls,12pt,onecolumn,oneside]{IEEEtran}
\usepackage{pifont}
\usepackage{times,amsmath,color,amssymb,graphicx,epsfig,cite,psfrag,subfigure}
\usepackage{mathrsfs}
\usepackage{amsfonts}
\usepackage{algorithm}
\usepackage{algorithmic}
\usepackage{multirow}
\long\def\symbolfootnote[#1]#2{\begingroup%
\def\thefootnote{\fnsymbol{footnote}}\footnote[#1]{#2}\endgroup}

\newtheorem{pro}{\rm{\underline{\textbf{Problem}}}}

\newtheorem{prop}{\rm{\underline{\textbf{Proposition}}}}

\newtheorem{thm}{\rm{\underline{\textbf{Theorem}}}}
\newtheorem{lem}{\rm{\underline{\textbf{Lemma}}}}

\begin{document}
%\title{Optimal Resource Allocation for Full-Duplex Wireless-Powered Communication Network}%\vspace{-1mm}
\title{Full-Duplex  Wireless-Powered Communication Network with Energy Causality}%\vspace{-1mm}
\author{~~~~~Xin Kang,~\IEEEmembership{Member,~IEEE}, ~Chin Keong Ho,~\IEEEmembership{Member,~IEEE}, \newline~Sumei
Sun,~\IEEEmembership{Senior Member,~IEEE},
\thanks{X. Kang, C. K. Ho, S. Sun are with Institute for Infocomm Research, 1 Fusionopolis Way,
$\#$21-01 Connexis, South Tower, Singapore 138632 (E-mail: \{xkang,
hock, sunsm\}@i2r.a-star.edu.sg).} } \markboth{X. Kang et. al,
``Full-Duplex Wireless-Powered Communication Network with Energy
Causality''} {} \maketitle

\begin{abstract}
In this paper, we consider a wireless communication network with a
full-duplex hybrid access point (HAP) and a set of wireless users
with energy harvesting capabilities. The HAP implements the
full-duplex through two antennas: one for broadcasting wireless
energy to users in the downlink and one for receiving independent
information from users via time-division-multiple-access (TDMA) in
the uplink at the same time. All users can continuously harvest
wireless power from the HAP until its transmission slot, i.e., the
energy causality constraint is modeled by assuming that energy
harvested in the future cannot be used for tranmission. Hence,
latter users' energy harvesting time is coupled with the
transmission time of previous users. Under this setup, we
investigate the sum-throughput maximization (STM) problem and the
total-time minimization (TTM) problem for the proposed multi-user
full-duplex wireless-powered network. The STM problem is proved to
be a convex optimization problem. The optimal solution strategy is
then obtained in closed-form expression, which can be computed with
linear complexity. It is also shown that the sum throughput is
non-decreasing with increasing of the number of users. For the TTM
problem, by exploiting the properties of the coupling constraints,
we propose a two-step algorithm to obtain an optimal solution. Then,
for each problem, two suboptimal solutions are proposed and
investigated. Finally, the effect of user scheduling on STM and TTM
are investigated through simulations. It is also shown that
different user scheduling strategies should be used for STM and TTM.
\end{abstract}
%\vspace{-5mm}

\newpage

\section{Introduction}
In conventional wireless networks, such as sensor networks and
cellular networks, wireless devices are powered by replaceable or
rechargeable batteries. The operation time of these battery-powered
devices are usually limited. Though replacing or recharging the
batteries periodically may be a viable option, it may be
inconvenient (for a sensor network with thousands of distributed
sensor nodes), dangerous (for the devices located in toxic
environments), or even impossible (for the medical sensors implanted
inside human bodies) to do so. In such situations, energy harvesting
\cite{ozel2011transmission,ho2012optimal}, with potential to provide
a perpetual power supply, becomes an attractive approach to prolong
these wireless networks' lifetime. Typical sources for energy
harvesting includes solar and wind. Recently ambient radio signal is
receiving much research attention as a new viable source for energy
harvesting, supported by the advantage that the wireless signals can
carry both energy as well as information
\cite{varshney2008transporting}.

Wireless power technologies have evolved significantly to make
wireless power transfer (WPT) for wireless applications a reality
\cite{Shinohara11}. Wireless power can be harvested from the
environment such as the TV broadcast signals \cite{liu2013ambient}.
In \cite{liu2013ambient}, a wireless peer-to-peer communication
system powered solely by ambient radio signals has already been
successfully implemented. Relying on an ambient power source,
however, poses uncertainty in the amount of energy that can be
harvested, and hence there is no guarantee on the minimum data rate.
WPT can also be achieved by using dedicated power transmitters, such
as in passive radio frequency identification (RFID) systems
\cite{finkenzeller2003rfid,smith2013wirelessly}. Conventional RFID
receivers, known as tags, cannot store the energy to be used in
future, hence limiting the potential application space. Recently,
more advanced RFID receivers have been prototyped that are able to
store the harvested energy, allowing the tags to perform sensing or
computation tasks even when the harvested energy is not currently
available \cite{smith2013wirelessly}. In addition, more efficient
wireless energy harvesting via WPT is believed to be in widespread
use in the near future due to the advances in antenna and circuit
designs.

%However, it poses lots of new challenges on the design of resource
%allocation strategies for the wireless communications networks. This
%is mainly due to the highly time-varying availability of the
%renewable energy. For instance, solar energy and wind energy may
%vary significantly over time and locations depending on the weather
%and the climate conditions. Thus, conventional transmit power
%constraints are not suitable to model communications devices with
%renewable energy. Instead, energy harvesting constraints, are
%introduced. With energy harvesting constraints, in every time slot,
%the transmitter is allowed to use at most the amount of harvested
%and stored energy currently available. In other words, the
%transmitter can not consume any energy that may be harvested in
%future slots.

For the above reasons, wireless powered communication networks
(WPCNs), in which wireless devices are powered only by WPT, becomes
a promising research topic \cite{grover2010shannon,
zhang2011mimo,liu2013wireless,zhou2012wireless,huang2012enabling,yang2013dynamic}.
In \cite{grover2010shannon}, the authors studied the tradeoff
between information rate and power transfer in a frequency selective
wireless system. In \cite{zhang2011mimo}, the authors proposed using
multiple antennas to achieve simultaneous wireless information and
power transfer for the emerging self-powered wireless networks. In
\cite{liu2013wireless}, the optimal power splitting between
information decoding and energy harvesting was derived to minimize
the outage probability. Then, in \cite{zhou2012wireless}, practical
receiver design for simultaneous information and power transfer was
studied. The architecture and deployment issues for enabling
wireless power transfer were investigated in
\cite{huang2012enabling}. In \cite{yang2013dynamic}, the authors
investigated how to improve energy beamforming efficiency  by
balancing the resource allocation between channel estimation and
wireless power transfer.

Typical WPCN networks are RFID networks or sensor networks, in which
the devices are usually low-power-consumption sensors with small
form factor \cite{finkenzeller2003rfid}. Such devices typically
store the harvested power in supercapacitors which have the
advantages of small form factor, fast charging cycle, and can
sustain many years of charging and discharging cycles
\cite{simjee2006everlast}, as compared to using rechargeable
batteries. However, supercapacitors suffer from high self-discharge
\cite{kaus2010modelling} and may not be able to store the harvested
energy long enough to be used for the next communication cycle,
which  may be after a few days or weeks depending on the
applications.

%These aforementioned works mainly focused on point-to-point wireless communications.
%
In this paper, to exploit the broadcast nature of the wireless
medium, we consider the use of WPT to support  multiple users
concurrently, using supercapacitor for storing the harvested energy.
To this end, we employ a hybrid access point (HAP) with dual
functions: to perform WPT to all users, and concurrently act as an
access point to collect the users' data. To support multiple users,
we employ time-division-multiple-access (TDMA) for the uplink
communications, but we allow concurrent downlink WPT to all users
even during the uplink communications. Each user continues to
harvest energy from the HAP until it performs uplink information
transmission. To account for the high self-discharge characteristic
of supercapacitor and potential long delay between any two
communication cycles, we assume the users cannot use the harvested
energy after its transmission slot, i.e., each user can only use all
the energy harvested so far before its transmission. Consequently,
latter users' can harvest more energy. This is consistent with the
concept of energy causality considered in
\cite{ho2012optimal,ozel2011transmission}, but the constraint is
imposed here due to the use of supercapacitor. The HAP is equipped
with two antennas to enable the concurrent downlink-WPT and
uplink-communication operations, respectively. We assume perfect
isolation between the two antennas, or the known WPT signal is
removed via analog and digital self-interference cancellation
\cite{bharadia2013full}, such that there is no WPT interference to
the uplink signals. In practice,  We refer to such a proposed WPCN
with a dual-function HAP as a full-duplex WPCN. We note that the
full-duplex concept here differs from conventional full-duplex
communications systems where the full-duplex is for uplink and
downlink information transmission only.

A closely related system model is considered in
\cite{ju2014throughput}. However, in \cite{ju2014throughput}, the
HAP is limited to having one single antenna. Thus, HAP cannot
perform WPT and uplink communication concurrently. Consequently, the
optimization for the time allocation problem studied in
\cite{ju2014throughput} is simpler, and the problem studied in this
paper is more interesting and challenging due to the coupling
between users's energy harvesting and transmission time. It is also
worth pointing out there are some existing literature on the
simultaneous wireless information and power transfer
\cite{fouladgar2012transfer,ng2013energy,nasir2013relaying,
Huang2013}. However, these works mainly focused on realizing
single-direction (i.e., either uplink or downlink) simultaneous WPT
and information transmission. This is different from our work where
bi-directional simultaneous WPT and information transmission is
enabled. %Besides, in our model, we respect the causality constraint
%for the harvested energy, i.e., the users can only use all the
%harvested energy before and not after its transmission.

The  contribution and main results of this paper are listed as
follows:
\begin{itemize}
\item We propose a new model to enable simultaneous downlink WPT and uplink information transmission for a multi-user wireless network by employing a full-duplex HAP.
We characterize two fundamental optimization problems for the
proposed full-duplex WPCN: (i). Sum-throughput maximization (STM),
i.e., maximize the total throughput of the proposed WPCN subject to
a total time constant $T$. (ii). Total-time minimization (TTM),
i.e., minimize the total charging and transmission time of the
proposed WPCN subject to the constraints that each user has certain
amount of data to send back to the HAP.
  \item For the STM problem, we rigorously prove that the formulated problem
is a convex optimization problem. By using convex optimization
techniques, the optimal time allocation solution is obtained in
closed-form. For convenience of computation, we propose a simple
algorithm with linear complexity to calculate the optimal time
allocation. We show that the sum throughput of the network is
non-decreasing with the increasing of the number of users despite
having the same total time constraint. We also show by simulations
that the users with low SNR should be scheduled to transmit first.
  \item For the TTM problem, we show that the optimal solution is in general
not unique. By exploiting the properties of the coupling
constraints, we propose a two-step algorithm to obtain an optimal
time allocation of the formulated problem. We also show by
simulations that the users with high SNR should be scheduled to
transmit first.

\item For both problems (STM and TTM), two suboptimal time allocation schemes are given. It is shown that optimization improve the system performance significantly.
\end{itemize}

The rest of this paper is organized as follows. In Section
\ref{Sec-SysModel}, we describe the WPCN system model and the
proposed full-duplex protocol. In Section \ref{Sec-STM}, we present
the problem formulation of the STM problem, and derive its the
optimal solution. In Section \ref{Sec-TTM}, we present the problem
formulation of the TTM problem, and propose a two-step algorithm to
obtain an optimal solution. In Section \ref{Sec-STA}, we propose two
suboptimal solutions for each problem (STM and TTM). In Section
\ref{Sec-NumericalResults}, numerical results are given to study the
performance of the proposed time allocation schemes. Section
\ref{Sec-Conclusions} concludes the paper.

%
%These aforementioned works mainly focused on point-to-point wireless
%communications. However, wireless transmission is in nature a
%multiuser communication setup. Thus, in this paper, we consider a
%wireless communication network with a hybrid access point (HAP) and
%multiple wireless users. The HAP is assumed to have two antennas:
%one dedicated for downlink wireless energy transfer and one
%dedicated for uplink wireless information receiving. This means HAP
%is a full duplex system. We assume perfect isolation between the two
%antennas, and hence no interference to the uplink signal is
%received.
%

%The main contribution of this paper is listed as follows:
%\begin{itemize}
%  \item We investigate the optimal time allocation to
%maximize the throughput of the proposed system subject to a total
%time constant. We rigorously prove that the formulated throughput
%maximization problem is a convex optimization problem.
%  \item  The optimal time allocation solution is
%obtained in closed-form by using convex optimization techniques. For
%convenience of computation, we also propose a simple algorithm with
%linear complexity to calculate the optimal time allocation.
%  \item  We show by simulations that
%the total throughput of the network increases with the number of
%users despite having the same total time constraint. We also show
%that the users with low SNR should be scheduled to transmit first.
%\end{itemize}

\section{System Model}\label{Sec-SysModel}
\begin{figure}[t]
        \centering
        \includegraphics*[width=12cm]{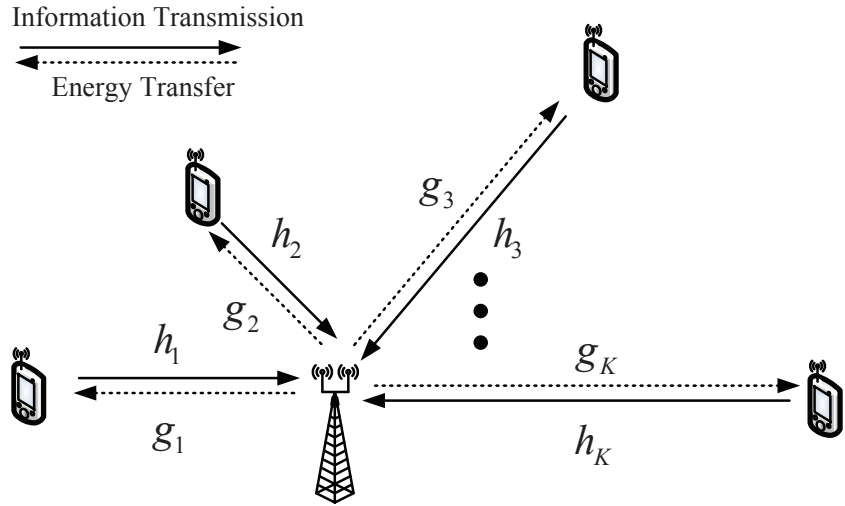}%\vspace{-3mm} %*[width=8cm]
        \caption{System Model: A Wireless Powered Communications Network}%\vspace{-3mm}
        \label{Fig-SystemModel}
\end{figure}
In this paper, we consider a WPCN with one HAP and $K$ users. The
HAP is assumed to be equipped with two antennas. One is for the
downlink wireless energy transfer. The other one is used to receive
the uplink information transmission from the users. All the user
terminals are assumed to be equipped with one single antenna each.
As shown in Fig. \ref{Fig-SystemModel}, the channel power gain of
the downlink channel from the HAP to user $i$ is denoted by $g_i$.
The channel power gain of the uplink channel from user $i$ to the
HAP is denoted by $h_i$. For the convenience of exposition, all the
channels involved are assumed to be block-fading \cite{kangTWC},
i.e., the channels remain constant during each transmission block,
but possibly change from one block to another. It is also assumed
that all these channel power gains are perfectly known at the HAP.

The frame structure for energy harvesting and information
transmission over one transmission block is shown in Fig.
\ref{Fig-EHModel}. In each block, the HAP keeps broadcasting
wireless energy to all the users with a constant transmit power
using one of its antenna. The transmit power of the HAP is denoted
by $P_H$. %It is also assumed that all the user terminals have no
%embedded energy sources. Thus, they have to first harvest energy
%from the signals broadcast by the HAP, and then use the harvested
%energy to transmit information back to the HAP.
To ensure tens of years of WPCN operations and small form factor for
the users, the harvested energy is stored in supercapacitors. Since
supercapacitors suffer from high
self-discharge\cite{kaus2010modelling}, we assume that the users can
harvest energy before its transmission but not after.
%All users can harvest the wireless power prior to its transmission.
Hence, latter users can harvest more energy. We assume the users
have no other energy source nor battery to store its harvested
energy, and hence all harvested energy must be used for transmission
with the frame. A TDMA structure is employed by the HAP to receive
the uplink information transmission from the users. For convenience,
we assume user $i$ transmits during the time slot $i$. The uplink
transmission time for user $i$ is denoted by $\tau_i$, $\forall i
\in \{1,\cdots,K\}$.

\begin{figure}[t]
        \centering
        \includegraphics*[width=12cm]{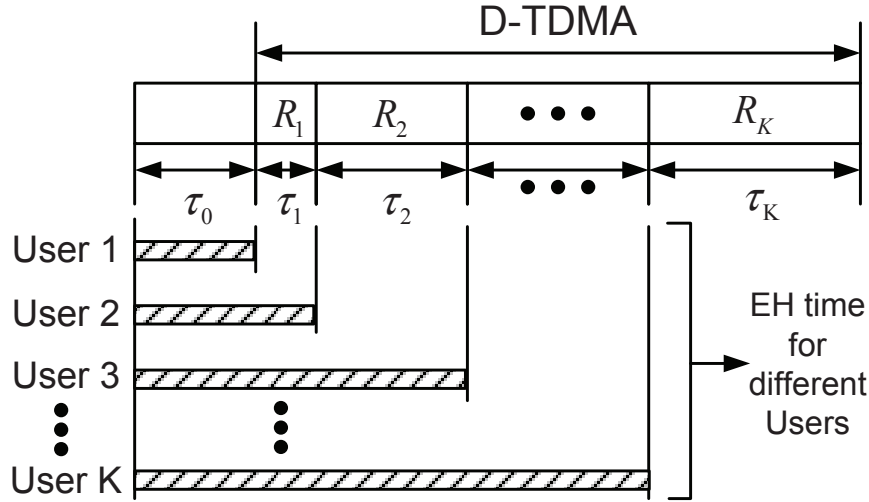}%\vspace{-3mm} %*[width=8cm]
        \caption{Frame structure for energy harvesting and information transmission. Latter users can harvest more energy. }%\vspace{-3mm}
        \label{Fig-EHModel}
\end{figure}

Then, the energy harvesting time of user $i$ is given by
$\sum_{j=0}^{i-1} \tau_j, ~\forall i \in\{1,2,\cdots,K\}$. Then, the
total energy harvested by user $i$ from the HAP, denoted as $E_i$,
can be obtained as
\begin{align}
E_i=\eta_i\sum_{j=0}^{i-1} g_{i}P_{H} \tau_j, ~\forall i
\in\{1,2,\cdots,K\},
\end{align}
where $\eta_i \in(0,1)$ is a constant denoting the energy harvesting
efficiency for user $i$.

The average transmit power $p_i$ for user $i$ during its
transmission slot is given by
\begin{align}
p_i=\frac{E_i}{\tau_i},~\forall i \in\{1,2,\cdots,K\}.
\end{align}

Since TDMA is employed, each user can only transmit during its
allocated time slot, and thus there is no mutual interference among
the users.  Besides, we assume that the HAP is equipped with a
successive interference cancellation decoder. Thus, the interference
from the downlink energy signal can be decoded first, and then be
subtracted from the received signal. Then, the instantaneous uplink
transmission rate for user $i$ can be written as
\begin{align}
R_i&=\ln\left(1+\frac{h_{i}E_i}{\tau_i\sigma^2}\right),\nonumber\\
&=\ln\left(1+\frac{\gamma_i\sum_{j=0}^{i-1}
 \tau_j}{\tau_i}\right), ~\forall i
\in\{1,2,\cdots,K\},
\end{align}
where $\gamma_i$ is defined as
$\gamma_i=\frac{h_{i}\eta_ig_{i}P_{H}}{\sigma^2}$, $\forall i
\in\{1,2,\cdots,K\}$, and $\sigma^2$ is the noise power at the HAP.

In this paper, we are interested in the following two problems:
%\begin{itemize}
%  \item Sum-throughput maximization, i.e., maximize the total
%throughput of the proposed WPCN subject to a total time constant
%$T$. This problem is of interest to the communication networks with
%fixed-length frames.
%  \item Total-time minimization, i.e., minimize the total
%charging and transmission time of the proposed WPCN subject to the
%constraints that each user has certain amount of data to send back
%to the HAP. This problem is of interest to the communication
%networks with variable-length frames.
%\end{itemize}
%These two problems are investigated in the following two sections,
%respectively.
%
(i). Sum-throughput maximization, i.e., maximize the total
throughput of the proposed WPCN subject to a total time constant
$T$. (ii). Total-time minimization, i.e., minimize the total
charging and transmission time of the proposed WPCN subject to the
constraints that each user has certain amount of data to send back
to the HAP. These two problems are investigated in the following two
sections, respectively.

% A strong
%motivation for us to propose such a system model is that: Typical
%WPCN networks are RFID networks or sensor networks, in which the
%devices are usually low-power-consumption sensors, which are small
%in size. These small sensors are usually powered by button batteries
%to collect data and to support the standby circuit power
%consumption. Sensors usually feedback back their collected data to a
%central controller (HAP) periodically, e.g., one day or one week.
%The data transmission component of these sensors are usually powered
%by the energy obtained from energy harvesting. Usually, rechargeable
%batteries are not their first option to harvest and store energy.
%This is due to the following reasons: (i). In most cases, sensors
%are required to be as small as possible. However, the size of
%rechargeable batteries is usually very large. (ii). The energy
%harvested from wireless power is usually very limited. The battery
%may drain out before the next feedback period due to the leakage,
%which is a waste of energy. (iii). The charging rate of the
%conventional rechargeable battery is very slow, and thus requires a
%long charging time. Thus, supercapaciter, which is smaller in size
%and can be charged and discharges quickly, is an ideal for serving
%as the harvest energy harvesting component of sensors. Thus, with
%these kinds of sensors, we can turn on the HAP and charge the
%sensors only when the HAP has to collect the sensing data from the
%sensors, and turn off the HAP after the data collection. In this
%way, we can realize a highly energy-efficient sensor network.

\section{Sum-throughput Maximization}\label{Sec-STM}

\subsection{Problem Formulation}\label{Sec-ProbFormulation}
Define $\boldsymbol{\tau}=[\tau_0,~ \cdots,~ \tau_K]^T$, the total
throughput denoted by $\mathcal {T}(\boldsymbol{\tau})$ of the
system is
\begin{align}
\mathcal {T}(\boldsymbol{\tau})&=\sum_{i=1}^K \tau_i R_i,\nonumber\\
&=\sum_{i=1}^K \tau_i \ln\left(1+\frac{\gamma_i\sum_{j=0}^{i-1}
 \tau_j}{\tau_i}\right),
\end{align}
where $\gamma_i$ is defined as
$\gamma_i=\frac{h_{i}\eta_ig_{i}P_{H}}{\sigma^2}$, $\forall i
\in\{1,2,\cdots,K\}$.

In this section, we are interested in finding the optimal time
allocation strategy to maximize the total throughput $\mathcal
{T}(\boldsymbol{\tau})$ of the described WPCN subject to a time
constant $T$. For convenience, we use a normalized unit block time,
i.e., $T=1$.  Thus, the throughput maximization problem can be
formulated as

\begin{pro}\label{Problem-ThroughputMaximization}
\begin{align}
\max_{\boldsymbol{\tau}} &\sum_{i=1}^K \tau_i
\ln\left(1+\frac{\gamma_i\sum_{j=0}^{i-1}
 \tau_j}{\tau_i}\right), \label{Eq-ThroughputMaximization-Obj}\\
\mbox{s.t.}~&~\tau_i \ge 0, ~\forall i \in\{0,1,2,\cdots,K\},\label{Eq-ThroughputMaximization-conng}\\
&\sum_{i=0}^K\tau_i \le 1.\label{Eq-ThroughputMaximization-con}
\end{align}
\end{pro}

Problem \ref{Problem-ThroughputMaximization} is a convex
optimization problem. To show this, we first present the following
lemma.

\begin{lem}\label{Lemma-concave}
The throughput function of user $i$ given by
$\mathcal{T}_i(\boldsymbol{\tau})\triangleq\tau_i
\ln\left(1+\frac{\gamma_i\sum_{j=0}^{i-1}
 \tau_j}{\tau_i}\right)$, $\forall i \in\left\{1, \cdots, K\right\}$, is a concave function of
 $\boldsymbol{\tau}=[\tau_0,~ \cdots,~ \tau_K]^T\succcurlyeq  \boldsymbol{0}$.
\end{lem}
\begin{proof}
According to \cite{Convexoptimization}, a function is concave if its
Hessian is negative semidefinite.  Thus, to show
$\mathcal{T}_i(\boldsymbol{\tau})$ is a concave function of
 $\boldsymbol{\tau}$, we have show that its Hessian is negative semidefinite. Denote the Hessian of
$\mathcal{T}_i(\boldsymbol{\tau})$ by $\boldsymbol{H}_i$ and denote
its element by $d_{m,n}^{(i)}$ at
%\begin{align}
%\boldsymbol{H}_i=\left[d_{m,n}^{(i)}\right], ~0\le m,n \le K,
%\end{align}
%where $d_{m,n}^{(i)}$ denotes the element of $\boldsymbol{H}_i$ at
the $m$th row and $n$th column. The diagonal entries of
$\boldsymbol{H}_i$, i.e., $m=n$, can be obtained as
\begin{align}
d_{m,m}^{(i)}=\left\{\begin{array}{cc}
                       -\frac{\gamma_i^2 \tau_i}{\left(\gamma_i\sum_{j=0}^{i-1}
 \tau_j+\tau_i\right)^2}, &~m<i, \\
                       -\frac{\gamma_i^2 (\sum_{j=0}^{i-1}
 \tau_j)^2}{\tau_i\left(\gamma_i\sum_{j=0}^{i-1}
 \tau_j+\tau_i\right)^2}, &~m=i,  \\
                       0, &~m>i.
                     \end{array}
\right.
\end{align}
The off-diagonal entries of $\boldsymbol{H}_i$ can be obtained as
\begin{align}
d_{m,n}^{(i)}=\left\{\begin{array}{cc}
                       -\frac{\gamma_i^2 \tau_i}{\left(\gamma_i\sum_{j=0}^{i-1}
 \tau_j+\tau_i\right)^2}, &~m<i~\mbox{and}~n<i, \\
                       0, &~m>i~\mbox{or}~n>i,\\
-\frac{\gamma_i^2 \sum_{j=0}^{i-1}
 \tau_j}{\left(\gamma_i\sum_{j=0}^{i-1}
 \tau_j+\tau_i\right)^2}, &~\mbox{otherwise}.
                     \end{array}
\right.
\end{align}
For any given real vector $\boldsymbol{v}=[v_0,~ \cdots,~ v_K]^T$,
it follows that
\begin{align}
\boldsymbol{v}^T\boldsymbol{H}_i\boldsymbol{v}&=-\frac{\gamma_i^2
}{\tau_i\left(\gamma_i\sum_{j=0}^{i-1}
 \tau_j+\tau_i\right)^2}\left(\tau_i\sum_{j=0}^{i-1}v_j-v_i\sum_{j=0}^{i-1}
 \tau_j\right)^2 \nonumber\\&{\le}0,
\end{align}
where the inequality follows from the fact that $\tau_i \ge 0$.
Thus, $\boldsymbol{H}_i$ is negative semidefinite. Lemma
\ref{Lemma-concave} is then proved.
\end{proof}

\begin{prop}
Problem \ref{Problem-ThroughputMaximization} is a convex
optimization problem.
\end{prop}
\begin{proof}
According to \cite{Convexoptimization}, a nonnegative weighted
summation of concave functions is concave.  Then, it follows from
Lemma \ref{Lemma-concave} that the objective function of Problem
\ref{Problem-ThroughputMaximization} given by
\eqref{Problem-ThroughputMaximization} is a concave function of
$\boldsymbol{\tau}$ since it is a summation of
$\mathcal{T}_i(\boldsymbol{\tau})$'s. Besides, all the constraints
of Problem \ref{Problem-ThroughputMaximization} are affine. Thus, it
is clear that Problem \ref{Problem-ThroughputMaximization} is a
convex optimization problem.
\end{proof}

Another important feature of Problem
\ref{Problem-ThroughputMaximization} is presented in the following
proposition.

\begin{prop}\label{prop-equal1}
The optimal time allocation $\boldsymbol{\tau}^*=[\tau_0^*,~\cdots,
~\tau_K^*]^T$ of Problem \ref{Problem-ThroughputMaximization} must
satisfy $\sum_{i=0}^K\tau_i^* =1$.
\end{prop}
\begin{proof}
This can be proved by contradiction. Suppose
$\boldsymbol{\tau}^\prime=[\tau_0^\prime,~\cdots, ~\tau_K^\prime]^T$
is an optimal solution of Problem
\ref{Problem-ThroughputMaximization}, and it satisfies that
$\sum_{i=0}^K\tau_i^\prime <1$. It follows that
$\tau_0^\prime<1-\sum_{i=1}^K\tau_i^\prime$. It is easy to verify
that the objective function given in
\eqref{Eq-ThroughputMaximization-Obj} is a monotonic increasing
function with respect to $\tau_0$. Thus, the value of
\eqref{Eq-ThroughputMaximization-Obj} under the vector
$[1-\sum_{i=1}^K\tau_i^\prime,~ \tau_1^\prime, ~\cdots,
~\tau_K^\prime]^T $ is larger than that under
$\boldsymbol{\tau}^\prime$. This contradicts with our presumption.
Thus, the optimal $\boldsymbol{\tau}^*$ must satisfy
$\sum_{i=0}^K\tau_i^* =1$.
\end{proof}

%\begin{pro}\label{Problem-ThroughputMaximization}
%\begin{align}
%\max_{\boldsymbol{\tau}} &~\mathcal {T}(\boldsymbol{\tau}), \label{Eq-ThroughputMaximization-Obj}\\
%\mbox{s.t.}~&~\boldsymbol{\tau}\succcurlyeq \boldsymbol{0},\\
%&~\boldsymbol{1}^T \boldsymbol{\tau} \le T.
%\end{align}
%\end{pro}

\subsection{Optimal Solution}
%Since we have proved that Problem
%\ref{Problem-ThroughputMaximization} is a convex optimization
%problem. Thus, the following convex optimization techniques can be
%used to solve  Problem \ref{Problem-ThroughputMaximization}.
In this subsection, we derive the optimal solution of Problem
\ref{Problem-ThroughputMaximization} using convex optimization
techniques.

The Lagrangian of Problem \ref{Problem-ThroughputMaximization} is
\begin{align}
\mathcal {L}\left(\boldsymbol{\tau},\lambda\right)=\sum_{i=1}^K
\tau_i \ln\left(1+\frac{\gamma_i\sum_{j=0}^{i-1}
 \tau_j}{\tau_i}\right)-\lambda\left(\sum_{i=0}^K\tau_i-1\right),
\end{align}
where $\lambda$ is the non-negative Lagrangian dual variable
associated with the constraint given in
\eqref{Eq-ThroughputMaximization-con}.

Then, the dual function of Problem
\ref{Problem-ThroughputMaximization} can be written as
\begin{align}
\mathcal {G}(\lambda)=\min_{\boldsymbol{\tau}\in\mathcal {S}}
\mathcal {L}\left(\boldsymbol{\tau},\lambda\right),
\end{align}
where $\mathcal {S}$ is the feasible set of $\boldsymbol{\tau}$
specified by the constraints \eqref{Eq-ThroughputMaximization-conng}
and \eqref{Eq-ThroughputMaximization-con}. It is observed that there
exists an $\boldsymbol{\tau}\in\mathcal {S}$ with all strict
positive element (i.e., $\tau_i>0, \forall i \in \{0,1,\cdots,K\} $)
satisfying $\sum_{i=0}^K\tau_i<1$. Thus, according to the Slater's
condition \cite{Convexoptimization}, strong duality holds for this
problem. Thus, Problem \ref{Problem-ThroughputMaximization} can be
solved by solving its Karush-Kuhn-Tucker (KKT) conditions, which are
given by
\begin{align}
\sum_{i=0}^K\tau^*_i &\le 1, \label{eq-KKT1}\\
\lambda^*\left(\sum_{i=0}^K\tau^*_i -1\right)&=0,\label{eq-KKT2}\\
\frac{\partial \mathcal
{L}\left(\boldsymbol{\tau},\lambda^*\right)}{\partial
\tau_i}\Big|_{\tau_i=\tau_i^*}&=0,~ \forall i \in
\{0,1,\cdots,K\}\label{eq-KKT3},
\end{align}
where $\tau_i^*, \forall i$ and $\lambda^*$ denote the optimal
primal and dual solutions of Problem
\ref{Problem-ThroughputMaximization}. %Since we have shown that
%\eqref{eq-KKT1} must hold with equality in Proposition
%\ref{prop-equal1}, with loss of generality,  we assume
%$\lambda^*>0$.
Then, from \eqref{eq-KKT3}, it follows that
\begin{align}
\sum_{k=1}^K \frac{\gamma_k}{\gamma_k
\frac{\sum_{j=0}^{k-1}\tau_j^*}{\tau_k^*}+1}=\lambda^*&,\label{eq-KKT3a}\\
\mathcal {B}_i
\left(\frac{\sum_{j=0}^{i-1}\tau_j^*}{\tau_i^*}\right)+\sum_{k=i+1}^{K}
\frac{\gamma_k}{\gamma_k
\frac{\sum_{j=0}^{k-1}\tau_j^*}{\tau_k^*}+1}=\lambda^*&,
\nonumber\\\forall i
\in \{1,\cdots,K-1\}&.\label{eq-KKT3b}\\
\mathcal {B}_K
\left(\frac{\sum_{j=0}^{i-1}\tau_j^*}{\tau_i^*}\right)=\lambda^*&,\label{eq-KKT3c}
\end{align}
where $\mathcal {B}_i (x)$ is defined as $ \mathcal {B}_i
(x)\triangleq \ln(1+\gamma_ix)-\frac{\gamma_ix}{1+\gamma_ix}$,
$\forall i \in \{1,\cdots,K\}$.

It is observed that the right hand sides of equations
\eqref{eq-KKT3a}-\eqref{eq-KKT3c} are the same. Thus, substituting
\eqref{eq-KKT3a} into \eqref{eq-KKT3b} and \eqref{eq-KKT3c}, we have
\allowdisplaybreaks
%\begin{align}
% \mathcal {B}_1
%\left(\frac{\tau_0^*}{\tau_1^*}\right)-
%\frac{\gamma_1}{\gamma_1 \frac{\tau_0^*}{\tau_1^*}+1}&=0,\\
%\mathcal {B}_2 \left(\frac{\sum_{j=0}^1 \tau_j^*}{\tau_2^*}\right)-
%\frac{\gamma_2}{\gamma_2 \frac{\sum_{j=0}^1
%\tau_j^*}{\tau_2^*}+1}&=\frac{\gamma_1}{\gamma_1
%\frac{\tau_0^*}{\tau_1^*}+1},\\
%&~\vdots\\
%\mathcal {B}_K \left(\frac{\sum_{j=0}^{K-1}
%\tau_j^*}{\tau_K^*}\right)-\kern-0.5mm \frac{\gamma_K}{\gamma_K
%\frac{\sum_{j=0}^{K-1} \tau_j^*}{\tau_K^*}+1}&=\sum_{k=1}^{K-1}
%\frac{\gamma_k}{\gamma_k
%\frac{\sum_{j=0}^{k-1}\tau_j^*}{\tau_k^*}+1}.
%\end{align}

\begin{align}
 \mathcal {B}_1
\left(\frac{\tau_0^*}{\tau_1^*}\right)-
\frac{\gamma_1}{\gamma_1 \frac{\tau_0^*}{\tau_1^*}+1}&=0,\label{eq-B1}\\
&~\vdots\\
\mathcal {B}_K \left(\frac{\sum_{j=0}^{K-1}
\tau_j^*}{\tau_K^*}\right)-\kern-0.5mm \frac{\gamma_K}{\gamma_K
\frac{\sum_{j=0}^{K-1} \tau_j^*}{\tau_K^*}+1}&=\sum_{k=1}^{K-1}
\frac{\gamma_k}{\gamma_k
\frac{\sum_{j=0}^{k-1}\tau_j^*}{\tau_k^*}+1}.\label{eq-BK}
\end{align}

Now, we denote $\frac{\sum_{j=0}^{i-1} \tau_j^*}{\tau_i^*}$ by
$x_i$, i.e.,
\begin{align}
x_i\triangleq \frac{\sum_{j=0}^{i-1} \tau_j^*}{\tau_i^*}, ~\forall i
\in \{1,\cdots,K\}.
\end{align}
We denote the right hand side of \eqref{eq-B1}-\eqref{eq-BK} by
$c_i$, i.e.,
\begin{align}
c_1&\triangleq 0,\\
c_i&\triangleq \sum_{k=1}^{i-1} \frac{\gamma_k}{\gamma_k x_k+1},
\forall i \in \{2,\cdots,K\}. \label{eq-ci}
\end{align}
For convenience, we introduce the following function
\begin{align}
\mathcal {F}_i (x_i)\triangleq \mathcal {B}_i
(x_i)-\frac{\gamma_i}{\gamma_ix_i+1},~\forall i \in \{1,\cdots,K\}.
\end{align}

Let $c_i\ge 0, \forall i \in \{1,\cdots,K\}$ be a series of
constants, the solution of $\mathcal {F}_i (x)=c_i$  denoted by
$x_i$ can be obtained as
\begin{align}\label{eq-xi}
x_i=\frac{1}{\gamma_i}\left(e^{\mathcal
{W}\left(\frac{\gamma_i-1}{e^{c_i+1}}\right)+c_i+1}-1\right),
\forall i \in \{1,\cdots,K\}.
\end{align}
where $\mathcal {W}\left(\cdot\right)$ is the Lambert W-Function
\cite{LambertW-function}.

%It is easy to observe that if we define $c_i$ by $c_1\triangleq 0$
%and $c_i\triangleq \sum_{k=1}^{i-1} {\gamma_k}\big/({\gamma_k
%\frac{\sum_{j=0}^{k-1}\tau_j^*}{\tau_k^*}+1}), \forall i \in
%\{2,\cdots,K\}$, $x_i$ gives the value of $\frac{\sum_{j=0}^{i-1}
%\tau_j^*}{\tau_i^*}$, i.e., $\frac{\sum_{j=0}^{i-1}
%\tau_j^*}{\tau_i^*}=x_i$, $\forall i \in \{1,\cdots,K\}$. Thus,
%$c_i$ can be rewritten as $\sum_{k=1}^{i-1} \frac{\gamma_k}{\gamma_k
%x_k+1}, \forall i \in \{2,\cdots,K\}$. This indicates that $c_i$
%only depends on the value of previous $\{x_1,\cdots, x_{i-1}\}$. It
%is observed from \eqref{eq-xi} that $x_1$ can be easily calculated
%since $c_1=0$. Thus, $c_2$ can be calculated with the obtained
%$x_1$. Then, $x_2$ can be obtained by \eqref{eq-xi}. Using the same
%approach, all the remaining $x_i$ can be obtained sequentially.

It is observed from \eqref{eq-xi} that we need $c_i$ to compute
$x_i$. When $i=1$, $x_1$ can be easily calculated since $c_1=0$. To
compute $x_2$, we need the value of $c_2$. It is observed from
\eqref{eq-ci} that $c_2$ can be easily computed if $x_1$ is known.
Thus, $x_2$ can be computed with the obtained $x_1$. Similarly, for
all other $i\ge 2$,  it is observed from \eqref{eq-ci} that $c_i$
only depends on the value of previous $\{x_1,\cdots, x_{i-1}\}$.
Thus, using the same approach, all the remaining $x_i$ can be
computed sequentially.

Now, we proceed to obtain the solution for $\tau_i^*, \forall i \in
\{1,\cdots,K\}$. Based on the fact that $\sum_{i=0}^K\tau^*_i = 1$
and $\frac{\sum_{j=0}^{i-1} \tau_j^*}{\tau_i^*}=x_i$, $\forall i \in
\{1,\cdots,K\}$, with the obtained value of $x_i$, the optimal
$\tau_i^*$ can be obtained as
\begin{align}
\tau_K^*&=\frac{1}{1+x_K},\\
\tau_{i}^*&=\frac{1-\sum_{j=i+1}^{K}\tau_j^*}{1+x_{i}},\forall i \in \{K-1,\cdots,1\},\\
\tau_{0}^*&=1-(\tau_K^*+\cdots+\tau_1^*), \label{eq-tau0}
\end{align}
where $x_i$ is given by \eqref{eq-xi}.

For convenience of computing the optimal time allocation, the
following algorithm is proposed.

\begin{algorithm}[htb]
\caption{Optimal Time Allocation Computation}
\label{alg:optimalTmax}
\begin{algorithmic}[1]

\STATE Initialize: $c_1=0$, $x_0=0$,

\FOR{i=1:K}

\STATE{$x_i=\frac{1}{\gamma_i}\left(e^{\mathcal
{W}\left(\frac{\gamma_i-1}{e^{c_i+1}}\right)+c_i+1}-1\right)$,}

\STATE{$c_{i+1}=\sum_{k=1}^{i} \frac{\gamma_k}{\gamma_k x_k+1}$,}

\ENDFOR

\STATE Compute $\tau_K$ by $\tau_K^*=\frac{1}{1+x_K}$.

\FOR{i=K-1:0}

\STATE{$\tau_{i}^*=\frac{1-\sum_{j=i+1}^{K}\tau_j^*}{1+x_{i}}$,}

\ENDFOR

\STATE{Output: $\tau_0^*, \cdots, \tau_K^*$.}

\end{algorithmic}
\end{algorithm}

It is observed that Algorithm  \ref{alg:optimalTmax} is a two-pass
algorithm: One pass for sequentially computing $x_i$ and the other
pass for computing $\tau_i$ in reverse order. Therefore, the
algorithm need not be completely rerun when extending from $K$ users
to $K+1$ users. Instead, we only need to rerun the second pass. This
is due to the fact that the value $x_{K+1}$ does not affects the
computation of $x_i$ with $i<K+1$, since the first pass is
sequential. This indicates that the proposed algorithm has a good
scalability.

Another interesting observation is that the sum throughput of the
proposed WPCN is non-decreasing with the increasing of the number of
users despite having the same total time constraint, which is given
in the following theorem.

\begin{thm}
The sum throughput of the proposed WPCN is non-decreasing with the
increasing of the number of users.
\end{thm}
\begin{proof}
Consider there are $K$ users in the network. The optimal time
allocation is $\boldsymbol{\tau}^*_{(K)}$ and the maximum throughput
is $\mathcal {T}^*_{(K)}$. Now, we introduce one more user to the
network, and recompute the optimal time allocation
$\boldsymbol{\tau}^\prime_{(K+1)}$ and the maximum throughput
$\mathcal {T}^\prime_{(K+1)}$. Now, we show that $\mathcal
{T}^\prime_{(K+1)}\ge \mathcal {T}^*_{(K)}$.

This can be proved by contradiction.  Suppose $\mathcal
{T}^\prime_{(K+1)}<\mathcal {T}^*_{(K)}$. Now, we consider the
following time allocation $\boldsymbol{\tau}^{\prime\prime}_{(K+1)}$
for the $K+1$ users case. We set the time allocation of the $K$ old
users using $\boldsymbol{\tau}^*_{(K)}$ and set the time allocation
of the new user using $0$. It is clear that under this time
allocation, the resultant throughput denoted by $\mathcal
{T}^{\prime\prime}_{(K+1)}$ is equal to $\mathcal {T}^*_{(K)}$. It
follows that $\mathcal {T}^{\prime\prime}_{(K+1)}>\mathcal
{T}^\prime_{(K+1)}$. This contradicts with our presumption. Thus, it
follows $\mathcal {T}^\prime_{(K+1)}\ge \mathcal {T}^*_{(K)}$.
\end{proof}

\section{Total-time Minimization} \label{Sec-TTM}

\subsection{Problem Formulation}\label{Sec-ProbFormulation}
Let $D_i$ be the minimum amount of information that user $i$ has to
send back to the HAP in each data collection cycle. Then, we have
the following constraints
\begin{align}
\tau_i\ln\left(1+\frac{\gamma_i\sum_{j=0}^{i-1}
\tau_j}{\tau_i}\right) \ge D_i, ~\forall i \in\{1,2,\cdots,K\}.
\end{align}
We assume $D_i>0$ for all users, otherwise the user with $D_i=0$ can
be omitted from the system to minimize the data collection time for
that user.

In this section, we are interested in minimizing the completion time
of charging and transmission of all users' data in the proposed
WPCN, i.e., $\sum_{i=1}^K \tau_i$, by determining the optimal time
allocation strategy. Mathematically, this time minimization problem
can be written as
\begin{pro}\label{Problem-tauminimization}
\begin{align}
\min_{\boldsymbol{\tau}} ~&\sum_{i=0}^K \tau_i,\\
\mbox{s. t.}~&~\tau_i \ge 0, ~\forall i \in\{0,1,2,\cdots,K\},\label{Problem-tauminimization-obj}\\
\tau_i&\ln\left(1+\frac{\gamma_i\sum_{j=0}^{i-1}
\tau_j}{\tau_i}\right) \ge D_i, ~\forall i \in\{1,2,\cdots,K\},
\label{Problem-tauminimization-con}
\end{align}
\end{pro}
where $\boldsymbol{\tau}$ is defined as
$\boldsymbol{\tau}=[\tau_0,\cdots,\tau_K]^T$.

For convenience of expression, we refer to the constraint specified
by \eqref{Problem-tauminimization-con} when $i=k$ as the $k$th
constraint. First we observe that the $K$th constraint holds with
equality.
\begin{prop}
The optimal time allocation $\boldsymbol{\tau}^*$ of Problem
\ref{Problem-tauminimization} must satisfy the $K$th constraint with
equality, i.e., $ \tau_K^*\ln\left(1+\frac{\gamma_K\sum_{j=0}^{K-1}
\tau_j^*}{\tau_K^*}\right)=D_K$.
\end{prop}
\begin{proof}
This can be proved by contradictory. Suppose an optimal solution
$\tilde{\boldsymbol{\tau}}=[\tilde{\tau}_0, \cdots,
\tilde{\tau}_K]^T$ satisfying $
\tau_K^*\ln\left(1+{\gamma_K\sum_{j=0}^{K-1}
\tau_j^*}/{\tau_K^*}\right)>D_K$. Now, we consider the function
$f(x)\triangleq x\log\left(1+c/x\right)$, where $c$ is a constant.
It can be verified that $f(x)$ is a monotonic increasing function
with respect to $x$ when $x>0$. Thus, by fixing $\tilde{\tau}_0,
\cdots, \tilde{\tau}_{K-1}$, we can always find a $\tau_K^\prime$
such that $ \tau_K^\prime\ln\left(1+{\gamma_K \sum_{j=0}^{K-1}
\tilde{\tau}_j}/{\tau_K^\prime}\right)=D_K$, and it is clear that
$\tau_K^\prime<\tilde{\tau}_K$. This contradicts with our
presumption that $\tilde{\boldsymbol{\tau}}$ is optimal. Thus, the
optimal solution must satisfy the  $K$th constraint with equality.
\end{proof}

We observe that each user's energy harvesting time is coupled with
the transmission time of all the users before him/her. It is the
coupling among these constraints that makes the problem difficult to
solve. Thus, to solve Problem \ref{Problem-tauminimization}, we
first investigate the properties related that to the constraints. To
this end, the constraints given in
\eqref{Problem-tauminimization-con} can be rewritten as
%\begin{align}
%\tau_0&\ge\frac{\tau_1}{\gamma_1}\left(e^{\frac{D_1}{\tau_1}}-1\right),\\
%\tau_0+\tau_1&\ge\frac{\tau_2}{\gamma_2}\left(e^{\frac{D_2}{\tau_2}}-1\right),\\
%&~~\vdots\\
%\tau_0+\cdots+\tau_{K-1}&\ge\frac{\tau_K}{\gamma_K}\left(e^{\frac{D_K}{\tau_K}}-1\right),
%\end{align}
%\begin{align}
%\tau_0&\ge V_1(\tau_1),\\
%\tau_0+\tau_1&\ge V_2(\tau_2),\\
%&~~\vdots\\
%\tau_0+\cdots+\tau_{K-1}&\ge V_K(\tau_K),\label{eq-conK}
%\end{align}
\begin{align}
\sum_{j=0}^{i-1}\tau_j&\ge V_i(\tau_i),~\forall i
\in\{1,2,\cdots,K\},\label{eq-conK}
\end{align}
where $V_i(\tau_i)$ is defined as
\begin{align}
V_i(\tau_i)\triangleq\frac{\tau_i}{\gamma_i}\left(e^{\frac{D_i}{\tau_i}}-1\right),~\forall
i \in\{1,2,\cdots,K\}.
\end{align}
Proposition states that $V_i(\cdot)$ is a strictly decreasing
function.

\begin{prop}\label{proposition-decreasingfn}
The function $V_i(\tau_i)$ is a strictly decreasing function with
respect to $\tau_i$.
\end{prop}
\begin{proof}
%The second order derivative of $v_i(\tau_i)$ is
%\begin{align}
%\frac{d^2v_i(\tau_i)}{d\tau_i^2}=\frac{D_i^2e^{D_i/\tau_i}}{\gamma_i
%\tau_i^3}\stackrel{a}{>}0,
%\end{align}
%where the inequality ``a'' follows from the fact that $\tau_i \ge
%0$. Thus, $v_i(\tau_i)$ is a convex function with respect to
%$\tau_i$.
First we note that $\tau_i>0$ so as to serve the strictly positive
$D_i$. The first order derivative of $V_i(\tau_i)$ is
\begin{align}
\frac{dV_i(\tau_i)}{d\tau_i}&=-\frac{(D_i-\tau_i)e^{D_i/\tau_i}+\tau_i}{\gamma_i\tau_i}\nonumber\\
&\stackrel{a}{<}-\frac{(D_i-\tau_i)+\tau_i}{\gamma_i\tau_i} <0
%&\stackrel{a}{=}-\frac{(D_i-\tau_i)\sum_{k=0}^{\infty}\frac{(D_i/\tau_i)^k}{k!}+\tau_i}{\gamma_i\tau_i}\nonumber\\
%&=-\frac{\sum_{k=2}^{\infty}\frac{D_i^k}{\tau_i^{k-1}}\left(\frac{1}{(k-1)!}-\frac{1}{k!}\right)}{\gamma_i\tau_i},\label{eq-1storderofvi}
\end{align}
where the inequality ``a'' follows from $e^x>1$ for $x>0$ and the
assumption that $D_i>0$.
%Hence, the function $V_i(\tau_i)$ is a monotonic decreasing function with respect to $\tau_i$.
\end{proof}

By using a graphical approach, the subsequent key result can be
proved easily and key insights can be drawn intuitively. Hence, we
illustrate the graphical relationship between the curve
$V_i(\tau_i)=\frac{\tau_i}{\gamma_i}\left(e^{\frac{D_i}{\tau_i}}-1\right)$
and the line $\sum_{j=0}^{i-1}\tau_j=-\tau_i+C_i$ in
Fig.~\ref{Fig-kxMinTimeDemoFig4}. We note that
$C_i=\sum_{j=0}^{i}\tau_j$ represents the completion time until user
$i$ sends its data; in particular $C_K$ is the total completion time
that we want to minimize. Lemma~\ref{lem:1} states the critical
values of $C_i$ and $\tau_i$ when the curve and line meet at one
unique point.

\begin{figure}%[t]
        \centering
        \includegraphics*[width=12cm]{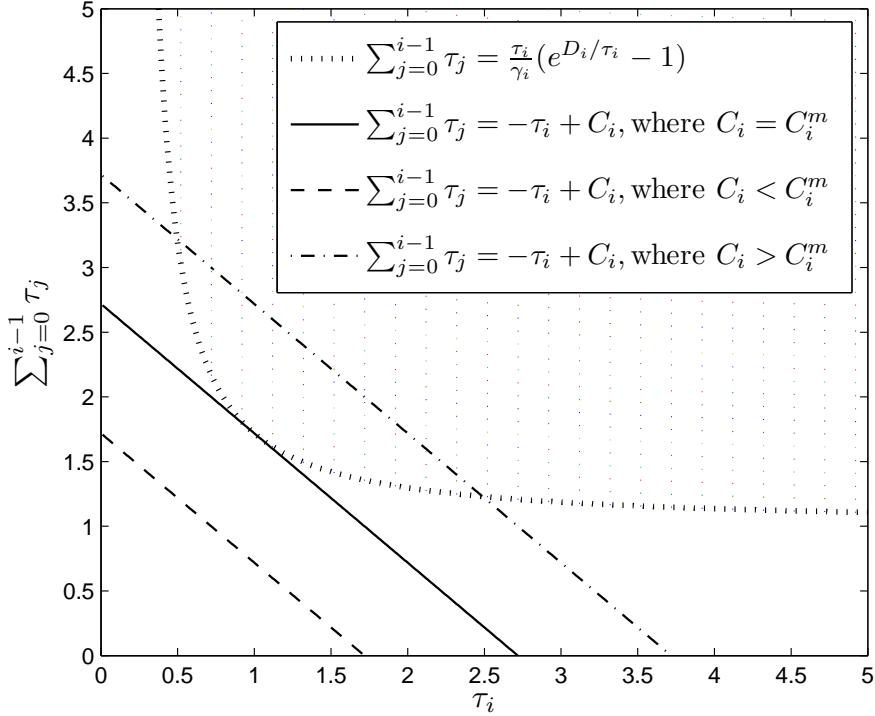}%\vspace{-3mm} %*[width=8cm]
        \caption{Graphic representation of the $i$th constraint}%\vspace{-3mm}
        \label{Fig-kxMinTimeDemoFig4}
\end{figure}

\begin{lem}\label{lem:1}
Denote $C_i^m$ as $C_i$ such that the line
$\sum_{j=0}^{i-1}\tau_j=-\tau_i+C_i$ has only one intersection point
with the curve $V_i(\tau_i)$. Denote the corresponding $\tau_i$ as
$\tau_i^m$.
Then %$\tau_i^m$ and $C_i^m$ can be obtained respectively as
\begin{align}\label{eq-tauim}
\tau_i^m&=\frac{D_i}{\mathcal
{W}\left(\frac{\gamma_i-1}{e}\right)+1} \\
C_i^m&=\frac{D_i }{\gamma_i}e^{\mathcal
{W}\left(\frac{\gamma_i-1}{e}\right)+1}. \label{eq-cim}
\end{align}
where $\mathcal {W}\left(\cdot\right)$ is the Lambert W-Function
\cite{kangTWC}.
\end{lem}
\begin{proof}
At the unique intersection point, we have
\begin{align}
\left. \frac{dV_i(\tau_i)}{d\tau_i}\right|_{\tau_i=\tau_i^m}&=-1,\label{eq-granegative1}\\
V_i(\tau_i^m)&=-\tau_i^m+C_i^m.\label{eq-yequal}
\end{align}
From \eqref{eq-granegative1}, it follows that
$-\frac{(D_i-\tau_i^m)e^{D_i/\tau_i^m}+\tau_i^m}{\gamma_i\tau_i^m}=-1$.
Then, $\tau_i^m$ can be obtained as given in \eqref{eq-tauim}. From
\eqref{eq-yequal}, it follows that $ C_i^m=\frac{D_i
}{\gamma_i}e^{\frac{D_i}{\tau_i^m}}$. Then, substituting
\eqref{eq-tauim} into this equation, $C_i^m$ can be obtained as
given in \eqref{eq-cim}.
\end{proof}

From Fig. \ref{Fig-kxMinTimeDemoFig4}, we have the following
observations:
\begin{itemize}
  \item When $C_i<C_i^m$, there is no intersection between the line and
the curve. When $C_i>C_i^m$, there are two intersection points
between the line and the curve.
  \item The shaded area is the feasible region specified by the $i$th constraint.
%  \item Since $C_i=\sum_{j=0}^i\tau_j$,  $\forall i=1,\cdots, K$.
% Thus, minimizing $\sum_{i=0}^K \tau_i$ is equivalent to minimizing $C_K$.
\end{itemize}

%(i). When $C_i<C_i^m$, there is no intersection between the line and
%the curve. (ii). When $C_i>C_i^m$, there are two intersection points
%between the line and the curve. (iii). The shaded area is the
%feasible region specified by the $i$th constraint. (iv). Since
%$C_i=\sum_{j=0}^i\tau_j$,  $\forall i=1,\cdots, K$. Thus, minimizing
%$\sum_{i=0}^K \tau_i$ is actually equivalent to minimizing $C_K$.

With these obtained properties and observations, we are now ready
for solving Problem \ref{Problem-tauminimization} optimally, which
is presented in the following section.

\subsection{Optimal Solution} In this subsection, we derive the optimal
solution of Problem \ref{Problem-tauminimization} based on the
results obtained in the previous subsection.

%\begin{thm}
%The optimal solution of Problem \ref{Problem-tauminimization} is
%obtained at the largest $k$ such that: when $C_k=C_k^m$ and
%$\tau_k=\tau_k^m$, the computed $\tau_{j-j}$ and $C_{k-j}$ satisfy
%\begin{align}
%C_{k-j}\ge C_{k-j}^m,~\forall j=1,\cdots,k,
%\end{align}
%where $C_k\triangleq\sum_{i=0}^k\tau_i$,  and $\tau_k^m$ and $C_k^m$
%are computed by \eqref{eq-tauim} and \eqref{eq-cim}, respectively.
%\end{thm}

\begin{thm}\label{theorem-opTime}
The optimal solution of Problem \ref{Problem-tauminimization} is
obtained at the largest $k$ such that: when $C_k=C_k^m$ and
$\tau_k=\tau_k^m$, there exists a time allocation $[\tau_0,\cdots,
\tau_{k-1}]^T$ that satisfies the $(k-j)$th constraints, $\forall
j=1,\cdots,k$.
\end{thm}

%\begin{figure}[t]
% \centering
% \subfigure[The $K$th constraint]{
%  \includegraphics[width=0.22\textwidth]{kxMinTimeDemoFigKv2}
%   \label{fig:subfig1}
%   }\hspace{0.04cm}
% \subfigure[The $\{K\kern-1mm-\kern-0.5mm1\kern-0.4mm\}$th constraint]{
%  \includegraphics[width=0.21\textwidth]{kxMinTimeDemoFigKn1v2}
%   \label{fig:subfig2}
%   }
%  \label{fig:subfigureExample}
% \caption[Optional caption for list of figures]{%
%Graphic representation of the $K$th and
%$\{K\kern-1mm-\kern-0.5mm1\kern-0.4mm\}$th constraints}
%\end{figure}

\begin{proof}
From the previous section, it is known that minimizing $\sum_{i=0}^K
\tau_i$ is equivalent to minimizing $C_K$. Then, we plot the line
$\sum_{i=0}^{K-1} \tau_i=-\tau_K+C_K$ and the curve
$\sum_{i=0}^{K-1}
\tau_i=\frac{\tau_K}{\gamma_K}\left(e^{\frac{D_K}{\tau_K}}-1\right)$
in Fig. \ref{fig:subfig1}. Since the curve represents the $K$th
constraint, the minimum feasible $C_K$ under the $K$th constraint is
$C_K=C_K^m$, and the corresponding $\tau_K$ is given by
$\tau_K=\tau_K^m$, where $\tau_K^m$ and  $C_K^m$  are computed by
\eqref{eq-tauim} and \eqref{eq-cim}, respectively. Then, $C_{K-1}$
can be computed by $C_{K-1}=C_K-\tau_K=C_K^m-\tau_K^m$. Now, we have
the following two possible cases:
\begin{itemize}
  \item Case 1: $C_{K-1}> C_{K-1}^m$. From Fig. \ref{fig:subfig2}, it is observed that when
  $C_{K-1}>C_{K-1}^m$, there are two intersection points between the line $\sum_{i=0}^{K-2}
\tau_i=-\tau_{K-1}+C_{K-1}$ and the curve $\sum_{i=0}^{K-2}
\tau_i=\frac{\tau_{K-1}}{\gamma_{K-1}}\left(\exp(D_{K-1}/\tau_{K-1})-1\right)$.
Denote the $\tau'$ and $\tau''$, where $\tau'<\tau''$, as the two
solutions for $\tau$, see Fig. \ref{fig:subfig2}. In this case, we
can choose any $\tau_{K-1}$ between $\tau'$ and $\tau''$, and yet
satisfy the $(K-1)$th constraint. However, from the perspective of
satisfying the remaining constraints, we should choose the smallest
possible $\tau_{K-1}$, i.e, choose $\tau_{K-1}=\tau'$, due to the
following. A smaller $\tau_{K-1}$ results in a larger
$\sum_{i=0}^{K-2} \tau_i$, since from
Proposition~\ref{proposition-decreasingfn}, $\sum_{i=0}^{K-2}
\tau_i=\frac{\tau_{K-1}}{\gamma_{K-1}}\left(\exp(D_{K-1}/\tau_{K-1})-1\right)$
is a decreasing function with respective to $\tau_{K-1}$. Since
$C_{K-2}=\sum_{i=0}^{K-2} \tau_i$, a larger $\sum_{i=0}^{K-2}
\tau_i$ implies a larger $C_{K-2}$. A larger $C_{K-2}$ results in
the most relaxed $(K-2)$th constraint, i.e., the largest possible
feasible region. By induction, this also resulted in the most
relaxed $(K-j)$th constraint for $j=3,\cdots, K-1$. %If all the
%remaining $C_{K-j}$ obtained by this approach, satisfies $C_{K-j}>
%C_{K-j}^m$, $\forall j=3,\cdots,K$. The optimal solution is obtained
%at $C_K=C_K^m$ and $\tau_K=\tau_K^m$.

  \item Case 2: $C_{K-1}<C_{K-1}^m$. From Fig. \ref{fig:subfig2}, it is observed that when $C_{K-1}<C_{K-1}^m$, there is no intersection between the line $\sum_{i=0}^{K-2}
\tau_i=-\tau_{K-1}+C_{K-1}$ and the curve $\sum_{i=0}^{K-2}
\tau_i=\frac{\tau_{K-1}}{\gamma_{K-1}}\left(\exp(D_{K-1}/\tau_{K-1})-1\right)$,
which means that the $(K-1)$th constraint is not satisfied.  In
order to satisfy the $(K-1)$th constraint, we must increase the
value of $C_{K-1}$ to $C_{K-1}^m$. Since $C_{K-1}=C_K-\tau_K$,  the
value of $C_{K-1}$ can be increased by increasing $C_K$ or
decreasing $\tau_K$. If we keep $C_K=C_K^m$ and only decrease the
value of $\tau_K$, the $K$th constraint will not be satisfied. Thus,
the value of $C_K$ must be increased. This indicates $C_K^m$ is no
longer a feasible solution of Problem \ref{Problem-tauminimization}.
Hence it is necessary  to set the tentative optimal solution at
$C_{K-1}=C_{K-1}^m$ and $\tau_{K-1}=\tau_{K-1}^m$.

%In this case, $C_K$ can be obtained by $C_K=C_{K-1}^m+\tau_K$, where
%$\tau_K$ is obtained by solving the equation
%$\frac{\tau_{K}}{\gamma_{K}}\left(e^{\frac{D_{K}}{\tau_{K}}}-1\right)=C_{K-1}^m$.
%The solution of this equation is unique since the curve
%$\frac{\tau_{K}}{\gamma_{K}}\left(e^{\frac{D_{K}}{\tau_{K}}}-1\right)$
%is a monotonic decreasing function, as illustrated in Fig.
%\ref{fig:subfig1}.
\end{itemize}

%\begin{figure}[t]
% \centering
% \subfigure[The $K$th constraint]{
%  \includegraphics[width=0.22\textwidth]{kxMinTimeDemoFigKv2}
%   \label{fig:subfig1}
%   }\hspace{0.04cm}
% \subfigure[The $\{K\kern-1mm-\kern-0.5mm1\kern-0.4mm\}$th constraint]{
%  \includegraphics[width=0.21\textwidth]{kxMinTimeDemoFigKn1v2}
%   \label{fig:subfig2}
%   }
%  \label{fig:subfigureExample}
% \caption[Optional caption for list of figures]{%
%Graphic representation of the $K$th and
%$\{K\kern-1mm-\kern-0.5mm1\kern-0.4mm\}$th constraints}
%\end{figure}

\begin{figure}%[t]
 \centering
 \subfigure[The $K$th constraint]{
  \includegraphics[width=0.45\textwidth]{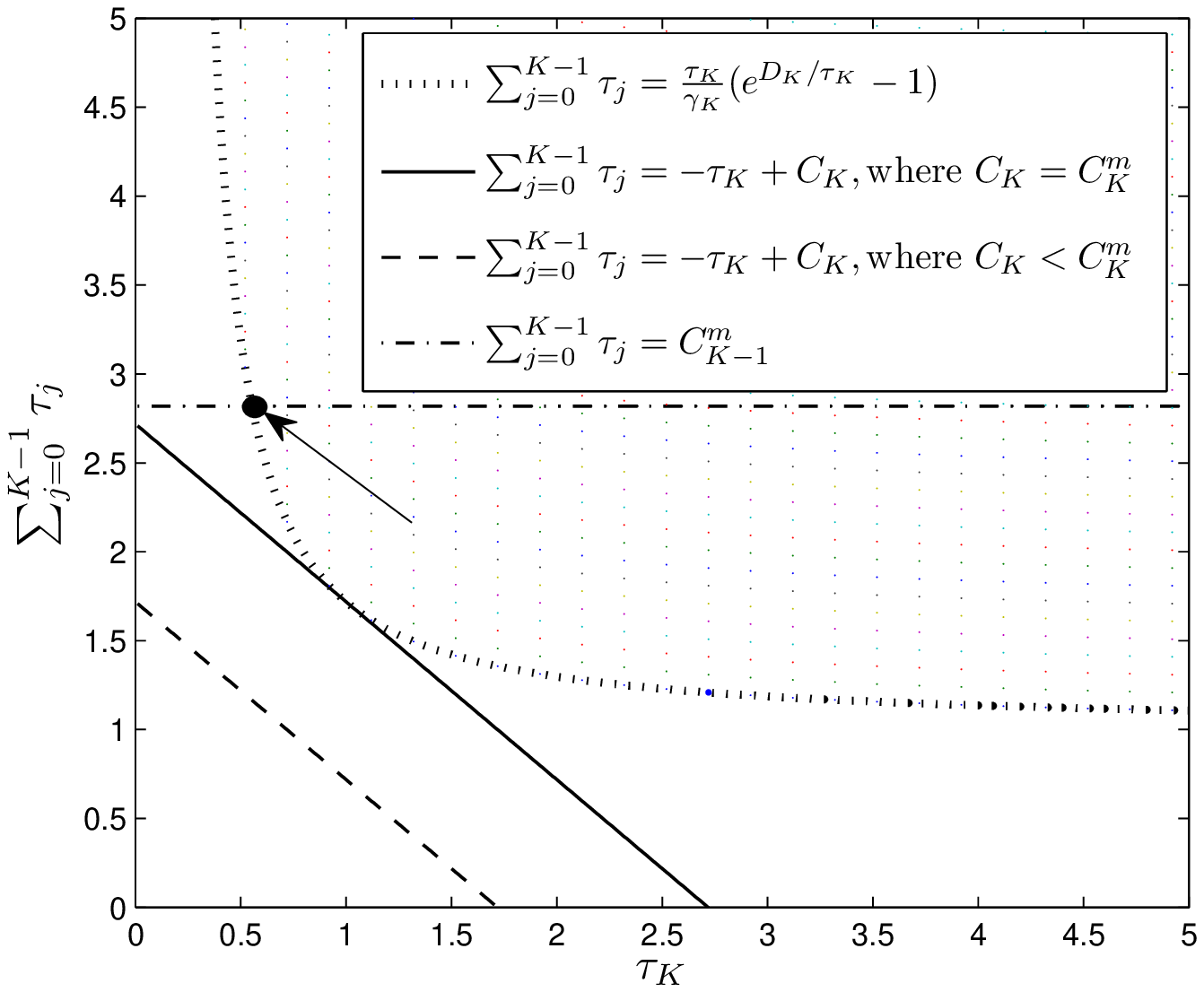}
   \label{fig:subfig1}
   }\hspace{0.01cm}
 \subfigure[The $(K\kern-1mm-\kern-0.5mm1\kern-0.4mm)$th constraint]{
  \includegraphics[width=0.45\textwidth]{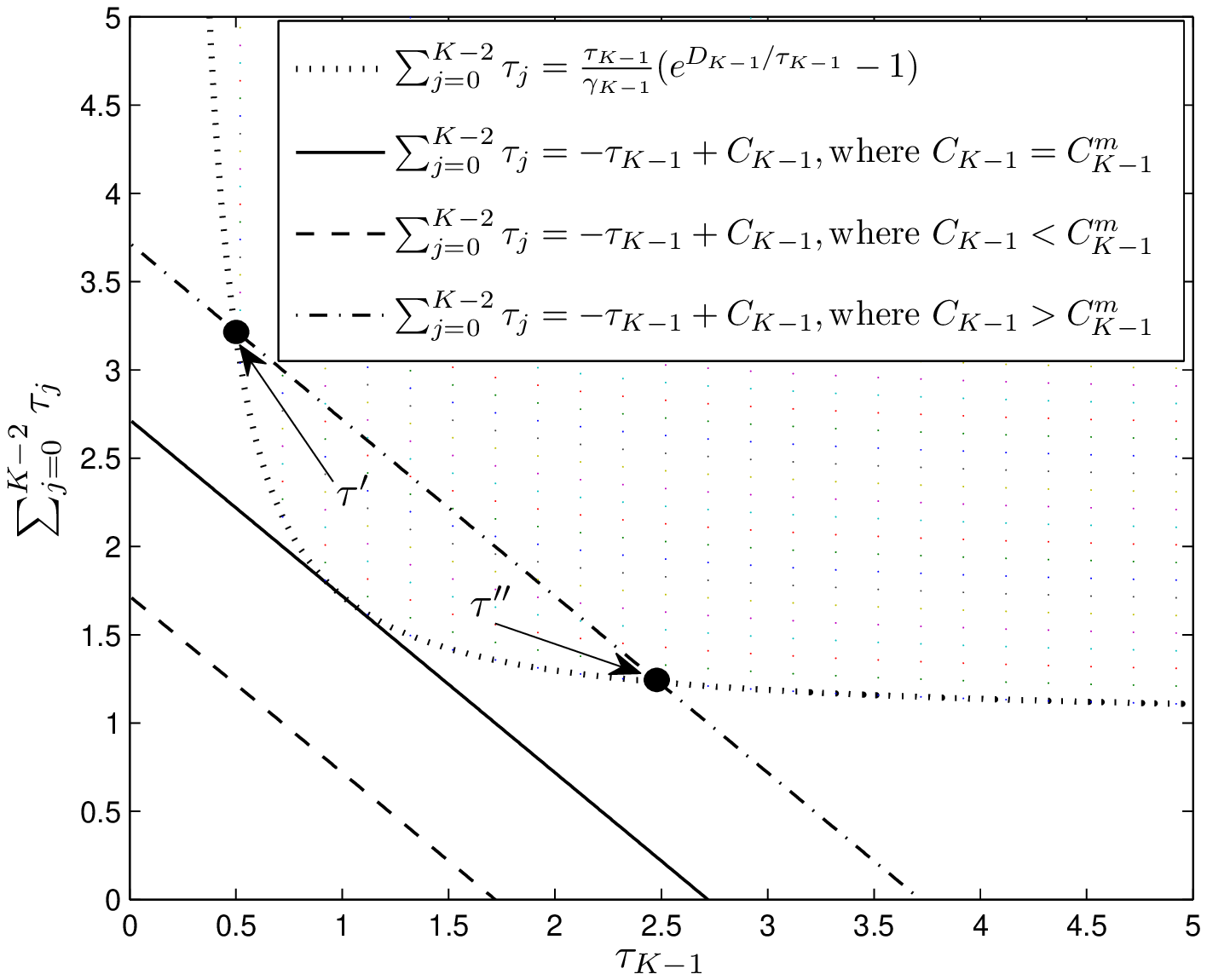}
   \label{fig:subfig2}
   }
  \label{fig:subfigureExample}
 \caption[Optional caption for list of figures]{%
Graphic representation of the $K$th and
$(K\kern-1mm-\kern-0.5mm1\kern-0.4mm)$th constraints}
\end{figure}

If case 1 happens, and all the remaining $C_{K-j}$'s computed by the
approach specified in case 1 satisfy $C_{K-j}> C_{K-j}^m$, $\forall
j=2,\cdots,K-1$, the optimal solution is obtained at $C_K=C_K^m$ and
$\tau_K=\tau_K^m$. If case 2 happens, we start from
$C_{K-1}=C_{K-1}^m$ and $\tau_{K-1}=\tau_{K-1}^m$, and recomputed
$C_{K-j}$, $\forall j=2,\cdots,K-1$.  If all the computed
$C_{K-j}$'s satisfy $C_{K-j}> C_{K-j}^m$, the optimal solution is
now obtained at $C_{K-1}=C_{K-1}^m$ and $\tau_{K-1}=\tau_{K-1}^m$.
Otherwise, we have to repeat this procedure until we find the
largest $k$ such that: when $C_k=C_k^m$ and $\tau_k=\tau_k^m$, all
the computed $C_{k-j}$'s satisfy $ C_{k-j}\ge C_{k-j}^m,~\forall
j=1,\cdots,k-1$.

%Clearly, if case 1 happens, the tentative optimal solution is obtained at
%$C_K=C_K^m$ and $\tau_K=\tau_K^m$. If case 2 happens, we have to
%repeat the procedure starting from $C_{K-1}=C_{K-1}^m$ and
%$\tau_{K-1}=\tau_{K-1}^m$ until we find the largest $k$ such that:
%when $C_k=C_k^m$ and $\tau_k=\tau_k^m$, all the $C_{k-j}$ computed
%using the approach in case 1, satisfies $ C_{k-j}\ge
%C_{k-j}^m,~\forall j=1,\cdots,k$.

Theorem \ref{theorem-opTime} is thus proved.
\end{proof}

Thus, the optimal solution of Problem \ref{Problem-tauminimization}
can be obtained in the following two steps: (i) Find the largest $k$
specified in Theorem \ref{theorem-opTime}. (ii) Let $C_k=C_k^m$ and
$\tau_k=\tau_k^m$, and then solve for the remaining the $C_i$'s and
$\tau_i$'s.

The pseudocode for finding the largest $k$ specified in Theorem
\ref{theorem-opTime} is given in Algorithm~\ref{alg:optimal}. In
Line~8 of Algorithm~\ref{alg:optimal}, the operator Root finds the
smaller of the two roots in the equation, i.e., $\tau_i'$ in
Fig.~\ref{fig:subfig2}.
\begin{algorithm}%[htb]
\caption{Finding the largest $k$ specified in Theorem
\ref{theorem-opTime}} \label{alg:optimal}
\begin{algorithmic}[1]

\STATE Initialize: $\mbox{Flag}=0$, $k=0$.

\STATE{Compute $\tau_{i}^m=\frac{D_i}{\mathcal
{W}\left(\frac{\gamma_i-1}{e}\right)+1}$, $C_i^m=\frac{D_i
}{\gamma_i}e^{\mathcal {W}\left(\frac{\gamma_i-1}{e}\right)+1}$,
$\forall i=1,\cdots,K$.}

\FOR{$j=K:-1:1$}

 \IF{$j>1$}

 \STATE{$C_{j-1}=C_j^m-\tau_{j}^m$;}

 \FOR{$i=j-1:-1:1$}

 \IF{$C_i\ge C_i^m$}

 \STATE{$\tau_{i}=\min\left\{ \mbox{Root}\left(\frac{\tau_i}{\gamma_i}\left(e^{\frac{D_i}{\tau_i}}-1\right)=-\tau_i+C_i\right)\right\}$;}

 \IF{$i>1$}

 \STATE{$C_{i-1}=C_i-\tau_i$;}

 \ELSE
 \STATE{$\mbox{Flag}=1$;}

\ENDIF

 \ELSE
 \STATE{$j=i$;}
 \STATE{$\mbox{break}$;}

\ENDIF

\ENDFOR

\ENDIF

 \IF{$\mbox{Flag}==1$}
 \STATE{Output: $k=j$.}
  \STATE{$\mbox{break}$;}
\ENDIF

 \STATE{Output: $k=1$.}
\ENDFOR

\end{algorithmic}
\end{algorithm}

Using Algorithm \ref{alg:optimal}, we can easily find the largest
$k$ specified by Theorem \ref{theorem-opTime}. Then, it follows that
$C_k^*=C_k^m$ and $\tau_k^*=\tau_k^m$. With this result, for any
$i<k$, $C_i$'s and $\tau_i$'s  can be easily obtained by using the
same method as Algorithm \ref{alg:optimal} (line $4$ to line $10$).
Now, we show how to compute $C_i$'s and $\tau_i$'s for any $i>k$. In
this case, $C_K$ can be obtained by $C_{i+1}^*=C_i^*+\tau_{i+1}^*$,
where $\tau_{i+1}^*$ is obtained by solving the equation
$\frac{\tau_{i+1}^*}{\gamma_{i+1}}\left(e^{\frac{D_{i+1}}{\tau_{i+1}^*}}-1\right)=C_{i}^*$.
The solution of this equation is unique since the curve
$\sum_{i=0}^i
\tau_i=\frac{\tau_{i+1}}{\gamma_{i+1}}\left(e^{\frac{D_{i+1}}{\tau_{i+1}}}-1\right)$
is a monotonic decreasing function of $\tau_{i+1}$, and
$\sum_{i=0}^i \tau_i=C_{i}$ is a horizontal line, which is
illustrated in Fig. \ref{fig:subfig1}. The approach to compute the
optimal time allocation presented above is summarized in the
following algorithm.

%The pseudo code for computing
%the optimal time allocation is then given in
%Algorithm~\ref{alg:optimalTime}.
\begin{algorithm}[htb]
\caption{Optimal time allocation computation}
\label{alg:optimalTime}
\begin{algorithmic}[1]

 \STATE{Input: $k$, which is obtained by Algorithm \ref{alg:optimal}.}

\STATE Initialize: $C_k^*=C_k^m$ and $\tau_k^*=\tau_k^m$.

\STATE{Compute $C_{k-1}^*=C_k^*-\tau_k^*$.}

\FOR{$i=k-1:-1:1$}

 \STATE{$\tau_{i}^*=\min\left\{ \mbox{Root}\left(\frac{\tau_i^*}{\gamma_i}\left(e^{\frac{D_i}{\tau_i^*}}-1\right)=-\tau_i^*+C_i^*\right)\right\}$;}

 \STATE{$C_{i-1}^*=C_i^*-\tau_i^*$.}

\ENDFOR

\FOR{$i=k+1:+1:K$}

 \STATE{$\tau_{i}^*=\mbox{Root}\left(\frac{\tau_i^*}{\gamma_i}\left(e^{\frac{D_i}{\tau_i^*}}-1\right)=C_{i-1}^*\right)$;}

 \STATE{$C_{i}^*=C_{i-1}^*+\tau_i^*$.}

\ENDFOR

 \STATE{$\tau_0^*=C_1^*-\tau_1^*$.}

  \STATE{Output: $\tau_0^*, \cdots, \tau_K^*$.}

\end{algorithmic}
\end{algorithm}

Algorithm \ref{alg:optimalTime} is an efficient way to obtain an
optimal solution of Problem \ref{Problem-tauminimization}. However,
it is worth pointing out that:
\begin{itemize}
  \item The optimal time allocation of Problem \ref{Problem-tauminimization}
may not be unique when $K\ge 2$. We use a simple example to
illustrate this. Consider the case that $K=2$. We assume that
$C_1=C_2-\tau_2>C_1^m$, when $C_2=C_2^m$ and $\tau_2=\tau_2^m$.
Then, any time allocation on the line $\tau_0=-\tau_1+C_1$ within
the feasible region $\tau_0\ge
\frac{\tau_1}{\gamma_1}\left(e^{\frac{D_1}{\tau_1}}-1\right)$ are
optimal time allocation.
  \item  When $K=1$, the optimal time allocation of Problem \ref{Problem-tauminimization}
is unique, and is given by $\tau_1=^*\frac{D_1}{\mathcal
{W}\left(\frac{\gamma_1-1}{e}\right)+1}$ and $
\tau_0^*=\frac{\tau_1^*}{\gamma_1}\left(e^{\frac{D_1}{\tau_1^*}}-1\right)$.
Details are omitted for brevity.
\end{itemize}

\section{Suboptimal Time Allocation} \label{Sec-STA}
In this section, we propose some suboptimal time allocation schemes
for the sum-throughput maximization and the total-time minimization
problems, respectively. The motivation for proposing these
suboptimal time allocation schemes is two-folds: (i). To see whether
optimization helps in improving the system performance. (ii). To
develop the low-complexity schemes that can achieve near-optimal
performance.

\subsection{Sum-throughput maximization}
In this subsection, we propose two suboptimal time allocation
schemes for the sum-throughput maximization problem, which are given
as follows.

(i). Equal time allocation. The idea of this scheme is to allocate
equal time to each user including the initial charging slot
$\tau_0$, i.e., $\tau_i=\tau_0$, $\forall i \in\{1,\cdots,K\}$.
Then, since $\sum_{i=0}^K\tau_i=1$, it follows that
\begin{align}\tau_i=\frac{1}{K+1}, \forall i
\in\{0,1,\cdots,K\}.\end{align}

(ii). Fixed-TDMA allocation:  The idea of this scheme is to allocate
equal time to each user but leave the initial charing slot $\tau_0$
as an optimization variable. Thus, it follows that
\begin{align}
\tau_i=\frac{1-\tau_0}{K}, \forall i \in\{1,\cdots,K\}.
\end{align}

Substituting the above equation into Problem
\ref{Problem-ThroughputMaximization}, it follows that
\begin{pro}\label{Problem-ThroughputMaximization3}
\begin{align}
\max_{\tau_0}&\sum_{i=1}^K \frac{1-\tau_0}{K}
\ln\left(i+\frac{K\gamma_i
 \tau_0}{1-\tau_0}\right), \label{Eq-ThroughputMaximization3-Obj}\\
\mbox{s.t.}~&~ 0\le\tau_0 \le
1,\label{Eq-ThroughputMaximization3-con}
\end{align}
\end{pro}

Now, we show that Problem \ref{Problem-ThroughputMaximization} is a
convex optimization problem. To show this, we present the following
lemma first.

\begin{lem}\label{Lemma-concave2}
The function of user $i$ given by
$\mathcal{T}^\prime_i(\tau_0)\triangleq\frac{1-\tau_0}{K}
\ln\left(i+\frac{K\gamma_i
 \tau_0}{1-\tau_0}\right)$, $\forall i \in\left\{1, \cdots, K\right\}$ is a concave function of
 $\tau_0$ when $0\le \tau_0 \le 1$.
\end{lem}
\begin{proof}
This can be proved by looking at the second order derivative of
$\mathcal{T}^\prime_i(\tau_0)$, which is
\begin{align}
\frac{d^2\mathcal{T}^\prime_i(\tau_0)}{dx^2}&=\frac{\gamma_i^2K}{(\tau_0-1)(\tau_0(\gamma_iK-i)+i)^2}\nonumber\\
&\stackrel{a}{\le}0,
\end{align}
where ``a'' follows from the fact that $0\le \tau_0 \le 1$. Thus, it
 is clear that $\mathcal{T}^\prime_i(\tau_0)$ is concave function of
 $\tau_0$ when $0\le \tau_0 \le 1$.
\end{proof}

Based on Lemma \ref{Lemma-concave2}, it is easy to observe that the
objective function of Problem \ref{Problem-ThroughputMaximization3}
is a concave function of $\tau_0$ when $0\le \tau_0 \le 1$, since
the summation operation preserves the concavity
\cite{Convexoptimization}. Since Problem
\ref{Problem-ThroughputMaximization3} is a convex optimization
problem with one optimization variable, it can be easily solved by
the subgradient method \cite{boyd2003subgradient}. Details are
omitted here for brevity.

%\begin{align}
%x_1&=\frac{\tau_0}{\tau_1},\\
%x_2&=\frac{\tau_0+\tau_1}{\tau_2}=1+x_1,\\
%&~\vdots\\
%x_K&=\frac{\tau_0+\cdots+\tau_{K-1}}{\tau_2}=K-1+x_1.
%\end{align}
%
%Once $x_1$ is known,
%\begin{align}
%\frac{\tau_0}{\frac{1-\tau_0}{K}}=x_1
%\end{align}
%
%\begin{align}
%\tau_0^*&=\frac{x_1}{K+x_1},\\
%\tau_i^*&=\frac{\tau_0^*}{x_1}=\frac{1}{K+x_1}, \forall i \in
%\left\{1,\cdots, K\right\}.
%\end{align}
%
%\begin{pro}\label{Problem-ThroughputMaximization2}
%\begin{align}
%\max_{\boldsymbol{\tau}} &\sum_{i=1}^K \tau_i
%\ln\left(1+\frac{\gamma_i\sum_{j=0}^{i-1}
% \tau_j}{\tau_i}\right), \label{Eq-ThroughputMaximization2-Obj}\\
%\mbox{s.t.}~&~\tau_i \ge 0, ~\forall i \in\{0,1,2,\cdots,K\},\label{Eq-ThroughputMaximization2-conng}\\
%&\sum_{i=0}^K\tau_i \le 1.\label{Eq-ThroughputMaximization2-con}\\
%&~\tau_1=\cdots=\tau_K.\label{Eq-ThroughputMaximization2-coneq}
%\end{align}
%\end{pro}

\subsection{Total-time minimization}
In this subsection, we propose two suboptimal time allocation
schemes for the total-time minimization problem,  which are given as
follows.

(i). Equal time allocation. The idea of this scheme is to allocate
equal time to each user including the initial charging slot
$\tau_0$, i.e., $\tau_i=\tau_0$, $\forall i \in\{1,\cdots,K\}$.
Substituting this condition into Problem
\eqref{Problem-tauminimization}, the problem is simplified as
\begin{pro}\label{Problem-tauminimization-subop}
\begin{align}
\min_{\boldsymbol{\tau}} ~& \tau_0,\\
\mbox{s. t.}~&~\tau_0\ln\left(1+i\gamma_i\right)\ge D_i, \forall i
\in\{1,\cdots,K\}.
\end{align}
\end{pro}
The optimal solution for this problem can be easily obtained, which
is
\begin{align}
\tau_0=\mbox{argmax}_i \frac{D_i}{\ln\left(1+i\gamma_i\right)}.
\end{align}

(ii). Tangent-point allocation: This scheme is inspired by the
graphic method used for deriving the optimal solution of Problem
\ref{Problem-tauminimization}. The idea of this scheme is to let
$\tau_i$ takes the value of the tangent-point illustrated in Fig.
\ref{Fig-kxMinTimeDemoFig4} for all $i$, i.e.,
\begin{align}
\widetilde{\tau}_i&=\frac{D_i}{\mathcal
{W}\left(\frac{\gamma_i-1}{e}\right)+1},\forall i=1,\cdots,K,
\end{align}
and $\widetilde{\tau}_0$ is given by the smallest value such that
all the constraints given in \eqref{Problem-tauminimization-con} are
satisfied. Since the left hand side of each constraint given in
\eqref{Problem-tauminimization-con} is a monotonic increasing
function with respect to $\tau_0$. Thus, $\tau_0$ can be easily
found by the well-known bisection search \cite{Convexoptimization}.
Details are omitted here for brevity.

\section{Numerical Results} \label{Sec-NumericalResults}
In this section, several numerical examples are presented to
evaluate the performance of the proposed algorithms.

\subsection{Simulation setup}
In the simulation, the power of the noise $\sigma^2$ at the receiver
of the HAP is assumed to be one. %For exposure, we assume that there
%are five users in the network, i.e., $K=5$.
For simplicity, the energy harvest efficiency for all users are
assumed to be the same and equal to one, i.e., $\eta_i=1, \forall
i$. The amount of data that each user has to send back is assumed to
be the same and equal to one, i.e., $D_i=1, \forall i$. We assume
i.i.d. Rayleigh fading for all channels involved, and thus the
channel power gains of these channels are exponentially distributed.
For convenience, we assume that the mean of the channel power gains
is one.  It is worth pointing out that the assumption of particular
distributions of the channel power gains does not affect the
structure of the problem studied and the algorithm proposed in this
paper. The results given in the following examples are obtained by
averaging over $1000$ channel realizations.
\begin{figure}[t]
        \centering
        \includegraphics*[width=12cm]{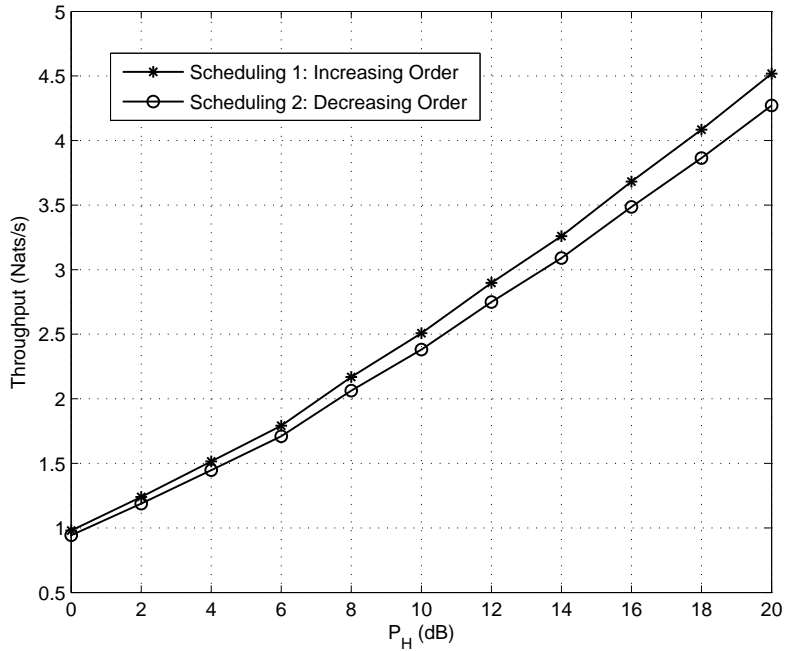}%\vspace{-3mm} %*[width=8cm]
        \caption{Throughput vs. HAP's transmit power ($P_H$)}%\vspace{-3mm}
        \label{Fig-kxfig3}
\end{figure}

\subsection{Sum-throughput maximization}

\subsubsection{Effect of the user scheduling}
In Fig. \ref{Fig-kxfig3}, we investigate the effect of user
scheduling on the throughput of the proposed system. We consider two
scheduling schemes: (i) Increasing order of SNRs, i.e., the user
with lowest SNR is scheduled to transmit first; (ii) Decreasing
order of SNRs, i.e., the user with highest SNR is scheduled to
transmit first. For exposure, we assume that there are five users in
the network, i.e., $K=5$. It is observed from Fig.~\ref{Fig-kxfig3}
that the increasing order scheduling scheme performs better than the
decreasing order scheduling scheme.  It is also observed that the
throughput gap between the two scheduling schemes increases with the
increasing of $P_H$. This indicates that user scheduling is more
important when $P_H$ is large.

\subsubsection{Optimal Vs. Suboptimal Time Allocation}
\begin{figure}[t]
        \centering
        \includegraphics*[width=12cm]{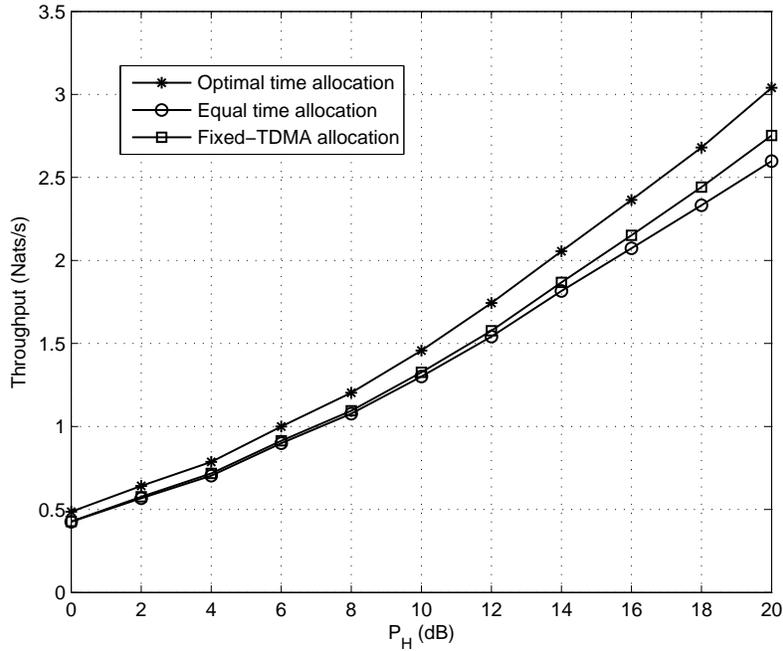}%\vspace{-3mm} %*[width=8cm]
        \caption{Throughput vs. HAP's transmit power ($P_H$)}%\vspace{-3mm}
        \label{Fig-kxfig1}
\end{figure}
In Fig. \ref{Fig-kxfig1}, we investigate the effect of the transmit
power of the HAP on the throughput of the proposed system under the
optimal time allocation and that under the equal time allocation.
 In this example, for simplicity, we consider the case
that there are two users, i.e., $K=2$. It is observed that the
optimal time allocation always perform better than the suboptimal
time allocation. It is also observed from Fig. \ref{Fig-kxfig1} that
the throughput for all cases increases when the transmit power of
the HAP ($P_H$) increases as expected. A higher $P_H$ indicates that
the users can harvest more energy from the HAP, and thus can
transmit at higher transmission rates. Therefore, the throughput of
the system increases. Another important observation is that the
throughput gap between the optimal time allocation and the equal
time allocation increases with the increasing of $P_H$. This
indicates that time allocation plays a more important role when
$P_H$ is large.

%\subsubsection{Effect of the number of users}
\begin{figure}[t]
        \centering
        \includegraphics*[width=12cm]{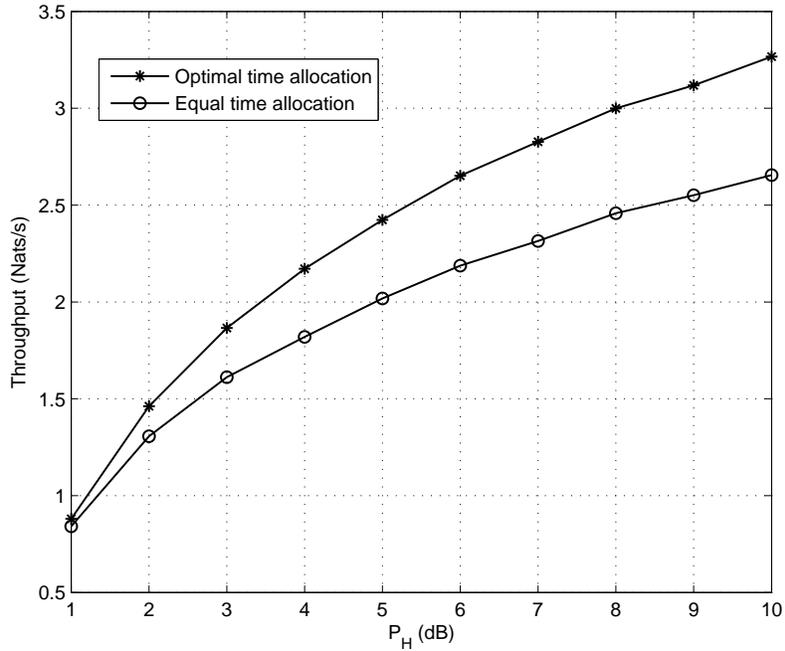}%\vspace{-3mm} %*[width=8cm]
        \caption{Throughput vs. Number of users ($K$)}%\vspace{-3mm}
        \label{Fig-kxfig2}
\end{figure}
In Fig. \ref{Fig-kxfig2}, we investigate the effect of the number of
users on the throughput of the proposed system under the optimal
time allocation and that under the equal time allocation.  In this
example, the transmit power of the HAP is fixed at $P_H=10~dB$. It
is observed that the throughput gap between the optimal time
allocation and the equal time allocation increases with the
increasing of the number of users ($K$). It is observed from Fig.
\ref{Fig-kxfig2} that when $K=1$, the gap is negligible. However,
when $K=10$, the gap is as large as $1$. This indicates that time
allocation plays a more important role when $K$ is large. Another
important observation is that the throughput for both cases
increases when the number of users increases. This is in accordance
to the results presented in Theorem 1. %The intuition is as
%follows. Suppose there are $5$ users. The optimal time allocation is
%$\boldsymbol{\tau}^*_{(K=5)}$ and the maximum throughput is
%$\mathcal {T}^*_{(K=5)}$. Now, we add $5$ more users to the system,
%and recompute the optimal time allocation
%$\boldsymbol{\tau}^\prime_{(K=10)}$ and the maximum throughput
%$\mathcal {T}^\prime_{(K=10)}$. If $\mathcal
%{T}^\prime_{(K=10)}=\mathcal {T}^*_{(K=5)}$, then the optimal
%solution for $K=10$ is setting the time allocation of the five old
%users using $\boldsymbol{\tau}^*_{(K=5)}$ and setting the time
%allocation of the five new users using $0$. Since the optimal time
%allocation for Problem \ref{Problem-ThroughputMaximization} must be
%positive (observed from \eqref{eq-xi}-\eqref{eq-tau0}), i.e.,
%$\tau_i>0,\forall i$. This contradicts with our presumption. Thus,
%$\mathcal {T}^\prime_{(K=10)}$ must be larger than $\mathcal
%{T}^*_{(K=5)}$.

\subsection{Total-time minimization}

\subsubsection{Effect of the user scheduling}
\begin{figure}[t]
        \centering
        \includegraphics*[width=10.5cm]{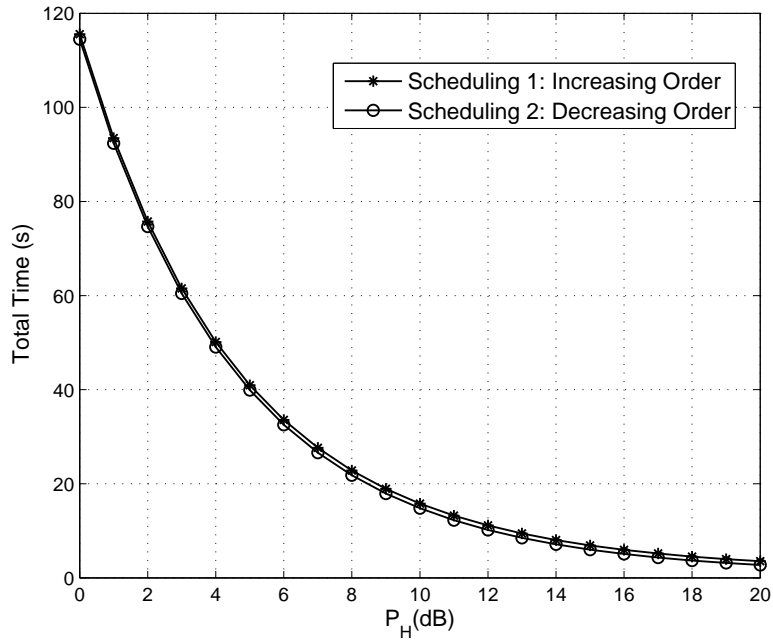}%\vspace{-3mm} %*[width=8cm]
        \caption{Total time vs. HAP's transmit power ($P_H$)}%\vspace{-3mm}
        \label{kxconf2fig2SU5}
\end{figure}
In Fig.~\ref{kxconf2fig2SU5}, we investigate the effect of the user
scheduling on the total data collection time of the proposed WPCN.
We consider two scheduling schemes: (i) Increasing order of SNRs,
i.e., the user with lowest SNR is scheduled to transmit first; (ii)
Decreasing order of SNRs, i.e., the user with highest SNR is
scheduled to
transmit first. % The results for $K=2$ and $K=5$ are given in Fig.
%\ref{kxconf2fig2Su2} and Fig. \ref{kxconf2fig2SU5}, respectively.
For exposure, we assume that there are five users in the network,
i.e., $K=5$. It is observed from Fig.~\ref{kxconf2fig2SU5} that the
increasing order scheduling scheme performs worse than the
decreasing order scheduling scheme, the gap between the two
scheduling schemes is, however, quite small. One of the possible
reasons is that users with good SNRs finish their data transmission
very fast, and thus contributes a negligible part of the total time.
On the other hand, users with the poor SNRs require long charging
and transmission time. Thus, the users with poor SNRs determine the
overall performance. Another observation is that the total time
decreases with the increasing of $P_H$. This is as expected, since
higher $P_H$ indicates that the users can harvest more energy from
the HAP, and thus can transmit at higher transmission rates.
%\begin{figure}[t]
% \centering
% \subfigure[$K=2$]{
%  \includegraphics[width=0.5\textwidth]{kxconf2fig2}
%   \label{kxconf2fig2Su2}
%   }\hspace{0.04cm}
% \subfigure[$K=5$]{
%  \includegraphics[width=0.5\textwidth]{kxconf2fig2SU5}
%   \label{kxconf2fig2SU5}
%   }
% \caption[Optional caption for list of figures]{%
%Total time vs. HAP's transmit power ($P_H$)}
%\end{figure}

%
%\begin{figure}[t]
%        \centering
%        \includegraphics*[width=8.5cm]{kxconf2fig2}%\vspace{-3mm} %*[width=8cm]
%        \caption{Throughput vs. HAP's transmit power ($P_H$)}%\vspace{-3mm}
%        \label{kxconf2fig2}
%\end{figure}
%
%\begin{figure}[t]
%        \centering
%        \includegraphics*[width=8.5cm]{kxconf2fig2SU5}%\vspace{-3mm} %*[width=8cm]
%        \caption{Throughput vs. HAP's transmit power ($P_H$)}%\vspace{-3mm}
%        \label{kxconf2fig2SU5}
%\end{figure}

\subsubsection{Optimal Vs. Suboptimal Time Allocation}
%\begin{figure}[t]
% \centering
% \subfigure[$K=2$]{
%  \includegraphics[width=0.5\textwidth]{kxconf2fig1}
%   \label{kxconf2fig1}
%   }\hspace{0.04cm}
% \subfigure[$K=5$]{
%  \includegraphics[width=0.5\textwidth]{kxconf2fig1SU5}
%   \label{kxconf2fig1SU5}
%   }
% \caption[Optional caption for list of figures]{%
%Optimal vs. Suboptimal time allocation}
%\end{figure}
In this subsection, we compare the performance of suboptimal time
allocation with the optimal time allocation.
%\begin{thm}
%A  suboptimal solution of Problem \ref{Problem-tauminimization} is
%given by
%\begin{align}
%\widetilde{\tau}_i&=\frac{D_i}{\mathcal
%{W}\left(\frac{\gamma_i-1}{e}\right)+1},\forall i=1,\cdots,K,
%\end{align}
%and $\widetilde{\tau}_0$ is given by the smallest value such that
%all the constraints given in \eqref{Problem-tauminimization-con} are
%satisfied. Since the left hand side of each constraint given in
%\eqref{Problem-tauminimization-con} is a monotonic increasing
%function with respect to $\tau_0$. Thus, $\tau_0$ can be easily
%found by the well-known bisection search.

In Fig. 6, we compare the performance of the proposed suboptimal
time allocation with that of the optimal time allocation.% The
%results for $K=2$ and $K=5$ are given in Fig. \ref{kxconf2fig1} and
%Fig. \ref{kxconf2fig1SU5}, respectively.
For exposure, we assume that there are five users in the network,
i.e., $K=5$. The equal time allocation is not included here due to
its poor performance. This also indicates that optimization is
necessary and helps in improving the system performance. As
expected, it is observed that the optimal time allocation always
perform better than tangent point allocation. It is also observed
that the gap between the optimal and the tangent point allocation is
very small. This indicates that the tangent point allocation scheme
has a very good performance. Another observation is that the total
time decreases with the increasing of $P_H$, and the gap also
decreases with the increasing of $P_H$.

\begin{figure}[t]
        \centering
        \includegraphics*[width=12cm]{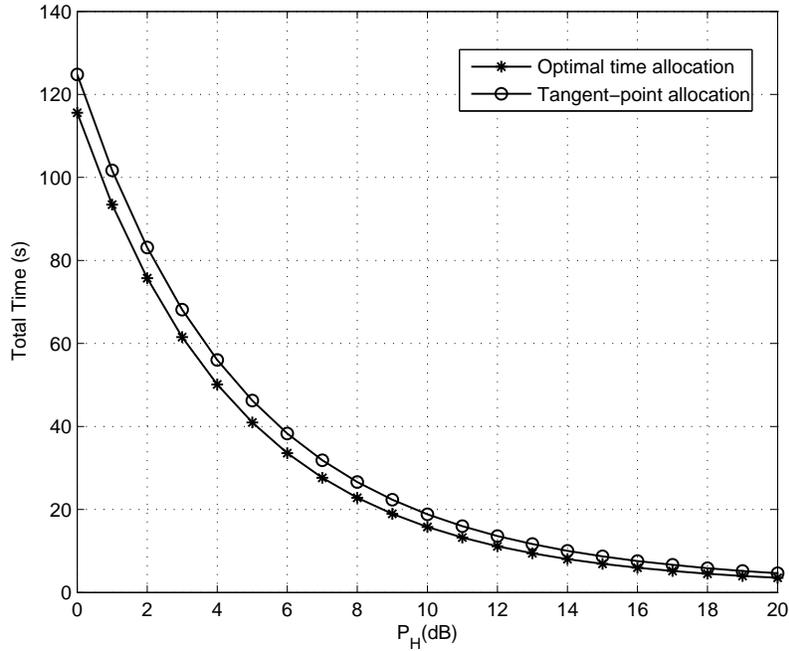}%\vspace{-3mm} %*[width=8cm]
        \caption{Optimal vs. Suboptimal time allocation}%\vspace{-3mm}
        \label{kxconf2fig1SU5}
\end{figure}

%\section{Numerical Results} \label{Sec-NumericalResults}
%In this section, several numerical examples are presented to
%evaluate the performance of the proposed algorithms.
%
%\subsection{Simulation setup}  In the simulation, the power $\sigma^2$ of the noise at the
%receiver of the HAP is assumed to be one. We assume i.i.d. Rayleigh
%fading for all channels,  and thus the channel power gains are
%exponentially distributed. We further assume that the mean of the
%channel power gains is one.  It is worth pointing out that the
%assumption of particular distributions of the channel power gains
%does not change the structure of the problem studied and the
%algorithm proposed in this paper. For simplicity, the energy harvest
%efficiency for all users are assumed to be the same and equal to
%one, i.e., $\eta_i=1, \forall i$. The results given in the following
%examples are obtained by averaging over $1000$ channel realizations.

\section{Conclusions}\label{Sec-Conclusions}
In this paper, we have proposed a new protocol to enable
simultaneous downlink wireless power transfer (WPT) and uplink
information transmission for a wireless communication network with a
full-duplex hybrid access point (HAP) and a set of wireless users
with energy harvesting capabilities.  Time-division-multiple-access
(TDMA) is employed to realize the multi-user uplink transmission.
All users can continuously harvest wireless power from the HAP until
its transmission slot, even during other users' uplink transmission.
Consequently, latter users' energy harvesting time is coupled with
the transmission time of previous users. Under this setup, we have
investigated the sum-throughput maximization (STM) problem and the
total-time minimization (TTM) problem for the proposed multi-user
full-duplex wireless-powered network, respectively. We have proved
that the STM problem is a convex optimization problem. The optimal
solution strategy has been obtained in closed-form expression. An
algorithm with linear complexity is then given for the convenience
of computation. We have shown that the sum-throughput is
non-decreasing with the increasing of the number of users. For the
TTM problem, we have proposed a two-step algorithm to obtain an
optimal solution by exploring the properties of the coupling
constraints. Then, we have proposed several suboptimal solutions are
each problem. We also have investigated the effect of user
scheduling on STM and TTM through simulations. We have shown that
different user scheduling strategies should be used for STM and TTM.

%\bibliographystyle{IEEEtran}
%\bibliography{EnergyHarvesting}

\end{document}